\documentclass[reqno, 11pt, twoside, a4]{amsart}
  \usepackage{multicurvestyle}
 
 \RequirePackage{colortbl}
 \definecolor{ts}{rgb}{0,0.6,0}
 \definecolor{cf}{rgb}{1.0,0.8,0.0}
 \definecolor{zg}{rgb}{0,0,0.6}
 \definecolor{sg}{rgb}{0.7,0,0.2}
 \definecolor{todo}{rgb}{0.2,0.2,0.3}
 
 \definecolor{refkey}{rgb}{0.9,0.9,0.9}
 \definecolor{labelkey}{rgb}{0.9,0.9,0.9}
 \usepackage{seqsplit}
 
\makeatletter
\newcommand{\addresseshere}{%
  \enddoc@text\let\enddoc@text\relax
}
\makeatother
 
\begin{document}

 	\title[Term structures with multiple curves and stochastic discontinuities]{Term structure modeling for multiple curves \\ with stochastic discontinuities}
 	
 	\author[C. Fontana, Z. Grbac, S. G{\"u}mbel \& T. Schmidt]{Claudio Fontana \and Zorana Grbac \and Sandrine G{\"u}mbel \and Thorsten Schmidt}
 	\address{University of Padova, 
 			Department of Mathematics ``Tullio Levi-Civita'', 
 			via Trieste 63, 35121, Padova, Italy.}
 	\email{fontana@math.unipd.it}
 	\address{Laboratoire de Probabilit\'es, Statistique et Mod\'elisation, Universit\'e Paris Diderot, avenue de France, 75205, Paris, France.}
 	\email{grbac@math.univ-paris-diderot.fr}
 	\address{University of Freiburg, Department of Mathematical Stochastics, Ernst-Zermelo-Str. 1, 79104 Freiburg, Germany.}
 	\email{sandrine.guembel@stochastik.uni-freiburg.de}
 	\address{Freiburg Institute of Advanced Studies (FRIAS), Germany. 
 University of Strasbourg Institute for Advanced Study (USIAS), France. 
 University of Freiburg, Department of Mathematical Stochastics, Ernst-Zermelo-Str. 1, 79104 Freiburg, Germany}
 	\email{thorsten.schmidt@stochastik.uni-freiburg.de}
 	
 	\thanks{The authors are thankful to two anonymous referees, an Associate Editor and the Editor for valuable remarks that helped improving the paper. The financial support from the Europlace Institute of Finance and the DFG project No. SCHM 2160/9-1 is gratefully acknowledged.}
 	\date{\today}
 	
 	
 	
 	\begin{abstract} 
 	    We develop a general term structure framework taking stochastic discontinuities explicitly into account. Stochastic discontinuities are a key feature in interest rate markets, as for example the jumps of the term structures in correspondence to monetary policy meetings of the ECB show. We provide a general analysis of multiple curve markets under minimal assumptions in an extended HJM framework and provide a fundamental theorem of asset pricing based on NAFLVR. The approach with stochastic discontinuities permits to embed market models directly, unifying seemingly different modeling philosophies. We also develop a tractable class of models, based on affine semimartingales, going beyond the requirement of stochastic continuity. 
 	\end{abstract}

 	\maketitle

\section{Introduction}

This work aims at providing a general analysis of interest rate markets in the post-crisis environment.
These markets exhibit two key characteristics. The first one is the presence of \emph{stochastic discontinuities}, meaning jumps occurring at predetermined dates. Indeed, a view on historical data of European reference interest rates  (see Figure \ref{fig:ECBratesEONIA}) shows surprisingly regular jumps: many of the jumps occur in correspondence of monetary policy meetings of the European Central Bank (ECB), and the latter take place at pre-scheduled dates. 
This important feature, present in interest rate markets even before the crisis, has been surprisingly neglected by existing stochastic models.

\begin{figure}[h!]
	\centering
	\includegraphics[width=\textwidth]{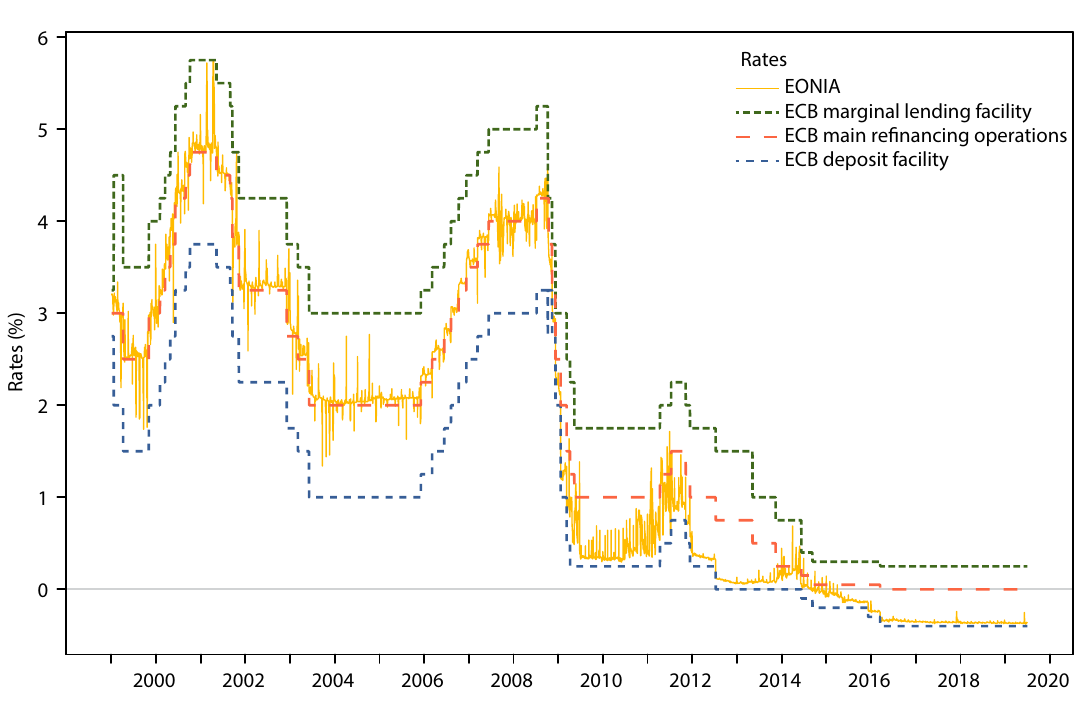}
	\caption{Historical series of the Eonia rate, the ECB deposit facility rate, the ECB marginal lending facility rate and the ECB main refinancing operations rate from January 1999 -- June 2019.	Source: European Central Bank.}
	\label{fig:ECBratesEONIA}
\end{figure}

The second key characteristic is the co-existence of different yield curves associated to different tenors. 
This phenomenon originated with the 2007 -- 2009 financial crisis,   when the spreads between different yield curves reached their peak beyond 200 basis points. Since then the spreads have remained on a non-negligible level, as shown in Figure \ref{fig:IborOISspread}. This was accompanied by a rapid development of interest rate models, treating multiple yield curves at different levels of generality and following different modeling paradigms. The most important curves to be considered in the current economic environment are the overnight indexed swap (OIS) rates and the interbank offered rates (abbreviated as Ibor, such as Libor rates from  the London interbank market) of various tenors. In the European market these are respectively the Eonia-based OIS rates and the Euribor rates.  

\begin{figure}[t]  
	\centering
	\includegraphics[width=\textwidth]{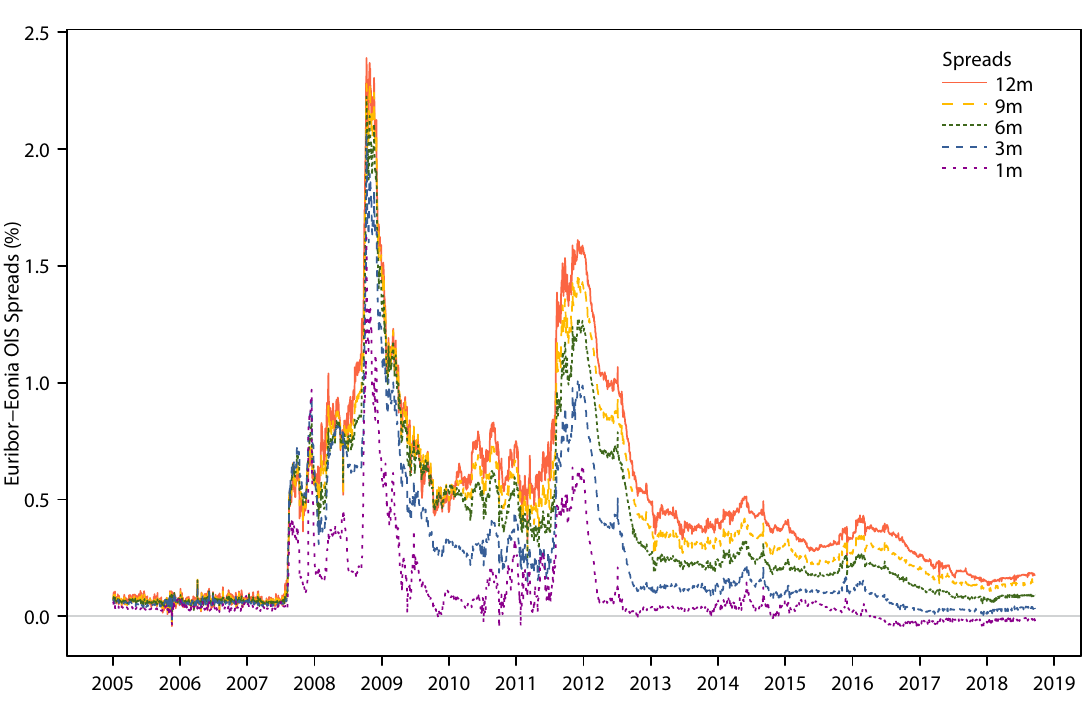}
	\caption{Euribor - Eonia OIS Spread for different maturities (1 month to 12 months) from  January 2005 -- September 2018. Source: Bloomberg and European Central Bank.}
	\label{fig:IborOISspread}
\end{figure} 	

It is our aim to propose a general treatment of markets with multiple yield curves in the light of stochastic discontinuities, meanwhile unifying the existing multiple curve modeling approaches.
The building blocks of this study are {\em OIS zero-coupon bonds} and \emph{forward rate agreements} (FRAs), which constitute the basic traded assets of a multiple curve financial market. 
While OIS bonds are bonds bootstrapped from quoted OIS rates, 
a FRA is an over-the-counter derivative consisting of an exchange of a payment based on a floating rate against a payment based on a fixed rate. FRAs can be regarded as the fundamental components of all interest rate derivatives written on Ibor rates.

The main contributions of the present paper can be outlined as follows:

\begin{itemize}
	\item
	A general forward rate setup for the term structure of FRAs and OIS bond prices inspired by the seminal Heath-Jarrow-Morton (HJM) approach of \cite{HJM}, suitably extended to allow for stochastic discontinuities. We derive a set of necessary and sufficient conditions characterizing risk-neutral measures  with respect to a general num\'eraire process (Theorem \ref{thm:HJM}). This framework unifies and generalizes the existing approaches in the literature.
	\item
	We study market models in general and, on the basis of minimal assumptions, derive necessary and sufficient drift conditions in the presence of stochastic discontinuities (Theorem \ref{thm:Libor}). This approach covers modeling under forward measures as a special case. Moreover, the generality of our forward rate formulation with stochastic discontinuities enables us to directly embed market models.
	\item
	We propose a new class of model specifications, based on affine semimartingales as recently introduced in \cite{KellerResselSchmidtWardenga2018}, going beyond the classical requirement of stochastic continuity. We illustrate the potential for practical applications by means of some simple examples.
	\item
	Finally, we provide a general description of a multiple curve financial market under minimal assumptions and a characterization of absence of arbitrage. We prove the equivalence between the notion of no asymptotic free lunch with vanishing risk (NAFLVR) and the existence of an equivalent separating measure (Theorem \ref{thm:FTAP}). To this effect, we rely on the theory of large financial markets and we extend to multiple curves and to an infinite time horizon the main result of \cite{CuchieroKleinTeichmann}. To the best of our knowledge, this represents the first rigorous formulation of an FTAP in the context of multiple curve financial markets.
\end{itemize}

\subsection{The modeling framework} 

We briefly illustrate the ingredients of our modeling framework, referring to the sections in the sequel for full details. First, forward rate agreements (FRAs) are quoted in terms of forward rates. More precisely, the 
\emph{forward Ibor rate} $L(t,T,\delta)$ at time $t\le T$ with {\em tenor} $\delta$ and settlement date $T$ is given as the unique value of the fixed rate which assigns the FRA value zero at inception $t$. This leads to the fundamental representation of FRA prices
\begin{equation} \label{eq:liborFRA}
	\pifra(t,T,\delta,K) = \delta \big( L(t,T,\delta) - K \big) P(t,T+\delta),
\end{equation}
where $P(t,T+\delta)$ is the price at time $t$ of an OIS zero-coupon bond with maturity $T+\delta$ and $K$ is a fixed rate.
Formula \eqref{eq:liborFRA} implicitly defines the yield curves  $T\mapsto L(t,T,\delta)$ for different tenors $\delta$, thus explaining the terminology multiple yield curves. In the following, we will simply call the associated markets \emph{multiple curve financial markets} (compare with Definition \ref{def:market} below).

The \emph{forward rate formulation} makes some additional assumptions on the yield curves. More specifically, it postulates that the right-hand side of \eqref{eq:liborFRA} admits the representation
\begin{equation}  \label{eq:FRAintro}
	\pifra(t,T,\delta,K) 
	= S_t^\delta  e^{-\int_{(t,T]} f(t,u,\delta) \eta(du) } - e ^{-\int_{(t,T+\delta]} f(t,u) \eta(du)} (1+\delta K).
\end{equation}
Here, $f(t,T)$ denotes the OIS forward rate, so that $P(t,T)=e ^{-\int_{(t,T]} f(t,u) \eta(du)}$, while $f(t,T,\delta)$ is the  $\delta$-tenor forward rate and $S^\delta $ is a multiplicative spread. We extend the usual HJM formulation by considering a measure $\eta$ containing atoms which by no-arbitrage will be precisely related to the set of stochastic discontinuities in the dynamics of forward rates and multiplicative spreads.

Representations \eqref{eq:liborFRA} and \eqref{eq:FRAintro}	constitute two seemingly different starting points for multiple curve modeling: {\em market models} and {\em HJM models}, respectively. In the following, we shall derive  no-arbitrage drift restrictions for both classes. Moreover, we will show that the two classes can be analyzed in a unified setting (see Appendix \ref{app:MM}).

\subsection{Stochastic discontinuities in interest rate markets} \label{sec:discontinuities}

One of the main novelties of our approach consists in the presence of stochastic discontinuities, representing events occurring at announced dates but with a possibly unanticipated informational content.
The importance of jumps at predetermined times is widely acknowledged in the financial literature, see for example \cite{Merton1974}, \cite{Piazzesi2001,Piazzesi2005,Piazzesi2010}, \cite{KimWright}, \cite{DuffieLando2001} (see also the introductory section of \cite{KellerResselSchmidtWardenga2018}). However, to the best of our knowledge, stochastic discontinuities have never been  taken explicitly into account in stochastic models for the term structure of interest rates. This feature is extremely relevant in financial markets. For instance, the Governing Council (GC) of the European Central Bank (ECB) holds its monetary policy meetings on a regular basis at predetermined dates, which are publicly known for about two years ahead. At such dates the GC takes its monetary policy decisions and determines whether the main ECB interest rates will change. In turn, these key interest rates are principal determinants of the Eonia rate, as illustrated by Figure \ref{fig:ECBratesEONIA}.

A closer inspection of Figure \ref{fig:ECBratesEONIA} reveals the presence of two different types of stochastic discontinuities in the Eonia rate. 
On the one hand, there are structural jumps in correspondence to monetary policy decisions. This type of discontinuity is evidenced by a step-like jump of the Eonia rate in correspondence to a new level of the ECB lending rate (see Figure \ref{fig:spikes}, right panel). 
On the other hand, there are spiky jumps which are unrelated to the monetary policy and occur at the end of the maintenance periods of banks' deposits. Indeed, in the Eurosystem banks are required to hold deposits on accounts with their national central bank over fixed maintenance periods. Banks who fail to keep sufficient reserves during the period need to borrow in the interbank market before the close of the maintenance period, thereby generating a temporary liquidity pressure in interbank lending which leads to a jump in the Eonia rate (see e.g. \cite{Beirne} and \cite{HT13}). This second type of stochastic discontinuity is evidenced by the spikes in the left panel of Figure \ref{fig:spikes}. 
More formally, we distinguish these two kinds of stochastic discontinuities as follows: jumps of \emph{type~I} are step-like jumps to a new level and jumps of \emph{type~II} are upward/downward jumps followed  by a fast continuous decay/ascent to the pre-jump level.

Our framework allows for the possibility of both type I and type II stochastic discontinuities.
In addition, by relaxing the classical assumption that the term structure of bond prices is absolutely continuous (see equation \eqref{eq:FRAintro}), we also allow for discontinuities in time-to-maturity at predetermined dates. 
In a credit risky setting, term structures with stochastic discontinuities have been recently studied in \cite{GehmlichSchmidt2016MF} and \cite{FontanaSchmidt2018}.
Finally, besides stochastic discontinuities as described above, we also allow for totally inaccessible jumps, representing events occurring as a surprise to the market and generated by a general random measure with absolutely continuous compensator. Such jumps have been already considered in several multiple curve models (see e.g. \cite{CrepeyGrbacNguyen12} and \cite{CuchieroFontanaGnoatto2016}).

\begin{figure}[t]
	\includegraphics[width=\textwidth]{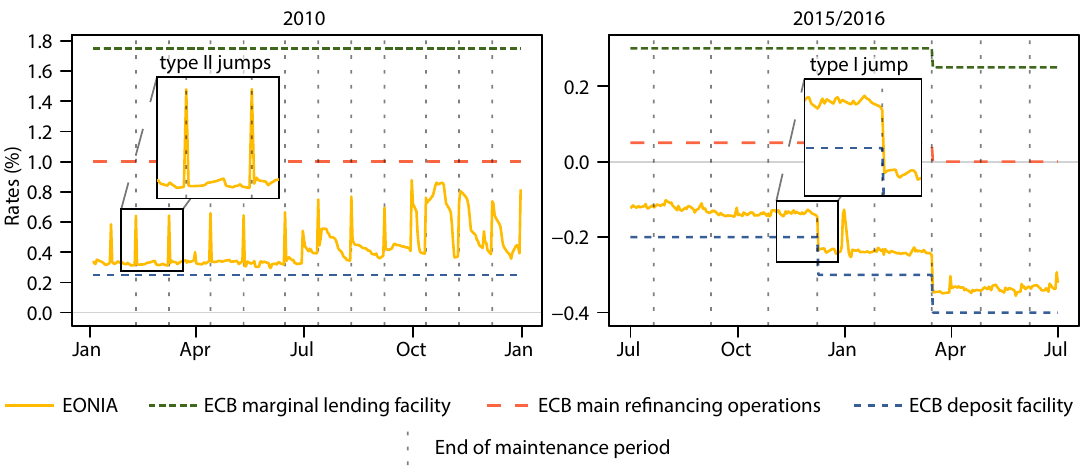}
	\caption{Eonia and ECB rates from January 2010 -- December 2010 (left panel) and  July 2015 -- June 2016 (right panel). The exposed discontinuities on the left panel are of type II, while the exposed discontinuity on the right panel is of type~I. Source: European Central Bank.}
	\label{fig:spikes}
\end{figure}

\subsection{Overview of the literature}
\label{literature-overview}

The literature on multiple curve models has witnessed a tremendous growth over the last years. Therefore, we only give an overview of the contributions that are the most related to the present paper, referring to the volume of \cite{BianchettiMorini13} and the monographs by \cite{Henrard2014} and \cite{GrbacRunggaldier} for further references and a guide to post-crisis interest rate markets. 
Multiplicative spreads for modeling multiple curves have been first considered in \cite{Henrard2007}. 
Adopting a short rate approach, an insightful empirical analysis has been conducted by \cite{FilipovicTrolle13}, showing that spreads can be decomposed into credit and liquidity components.
The extended HJM approach developed in Section \ref{sec:TSM} generalizes the framework of \cite{CuchieroFontanaGnoatto2016}, who consider It\^o semimartingales as driving processes and,  therefore, do not allow for stochastic discontinuities (see Remark \ref{rem:CFG} for a detailed comparison). 
HJM models taking into account multiple curves have been proposed in \cite{Crepey2015}  with L\'evy processes as drivers and in \cite{MoreniPallavicini14} in a Gaussian framework. 
In the market model setup, the extension to multiple curves was pioneered by \cite{Mercurio2010} and further developed in \cite{MercurioXie2012}. More recently, \cite{Grbac2015affine} have developed an affine market model in a forward rate setting, further generalized by \cite{CuchieroFontanaGnoatto2016b}.
All these models, both HJM and market models, can be easily embedded in the general framework proposed in this paper.

\subsection{Outline of the paper}

In Section \ref{sec:generalview}, we introduce the basic assets in a multiple curve financial market.  The general multi-curve framework inspired by the HJM philosophy, extended to allow for stochastic discontinuities, is developed and fully characterized in Section \ref{sec:TSM}. In Section \ref{sec:market models}, we introduce and analyze general market models with multiple curves. 
In Section \ref{sec:affine}, we propose a flexible class of models based on affine semimartingales, in a setup which allows for stochastic discontinuities.
In Section \ref{sec:NAFLVR}, we prove a version of the fundamental theorem of asset pricing for multiple curve financial markets, by relying on the theory of large financial markets. 
Finally, the two appendices contain some technical results and a result on the embedding of market models into the extended HJM framework.

\section{A general analysis of multiple curve financial markets}
\label{sec:generalview}

In this section, we provide a general description of a multiple curve market under minimal assumptions.
We assume that the interbank offered rates (Ibor) are quoted for a finite set of tenors $\cD:=\{\delta_1,\dots,\delta_m\}$, with $0< \delta_1 < \ldots < \delta_m$. Typically,  about seven tenors, ranging from 1 day to 12 months, are available in the market. 	
For a tenor $\delta\in\cD$, the Ibor rate for the time interval $[T, T+\delta]$ fixed at time $T$ is denoted by $L(T, T, \delta)$.  		
For $0\leq t\leq T<+\infty$, we denote by $P(t, T)$ the price at date $t$ of an OIS zero-coupon bond with maturity $T$. 

\begin{definition}
	A {\em forward rate agreement} (FRA) with tenor $\delta$, settlement date $T$, strike $K$ and unitary notional amount, is a contract in which a payment based on the Ibor rate $L(T,T,\delta)$ is exchanged against a payment based on the fixed rate $K$ at maturity $T+\delta$. The price of  a FRA contract at date $t\le T+\delta$ is denoted by $\pifra(t, T, \delta, K)$ and the payoff at maturity $T+\delta$ is given by
	\begin{align}	\label{eq:payoff_FRA}
		\pifra(T+\delta,T,\delta,K) 
		= \delta L(T,T,\delta) - \delta K.
	\end{align}
\end{definition}

The two addends in \eqref{eq:payoff_FRA} are typically referred to as floating leg and fixed leg, respectively. We define the multiple curve financial market as follows.

\begin{definition}	\label{def:market}
	The {\em multiple curve financial market} is the financial market containing the following two sets of traded assets:
	\begin{enumerate}
		\item[(i)] 
		OIS zero-coupon bonds, for all maturities $T\ge 0$;
		\item[(ii)] 
		FRAs, for all tenors $\delta\in\cD$, all settlement dates $T\ge 0$ and all strikes $K\in\R$.
	\end{enumerate}
\end{definition}

The assets included in Definition \ref{def:market} represent the quantities that we assume to be {\em tradable} in the financial market.
We emphasize that, in the post-crisis environment, FRA contracts have to be considered on top of OIS bonds as they cannot be perfectly replicated by the latter, due to the risks implicit in interbank transactions.

We work under the {\em standing assumption} that FRA prices are determined by a linear valuation functional. This assumption is standard in interest rate modeling and is also coherent with the fact that we consider {\em clean prices}, i.e., prices which do not model explicitly counterparty and liquidity risk (the counterparty and liquidity risk of the interbank market as a whole is of course present in Ibor rates, recall Figure \ref{fig:IborOISspread}). Clean prices are fundamental quantities in interest rate derivative valuation and they also form the basis for the computation of XVA adjustments, see Section 1.2.3 in \cite{GrbacRunggaldier} and \cite{Brigo_et_al2018}.

Recalling \eqref{eq:payoff_FRA}, the value of the fixed leg of a FRA at time $t\le T+\delta$ is given by $\delta K P(t,T+\delta)$. 
Hence, we obtain that $\pifra(t,T,\delta,K)$ is an affine function of $K$.

\begin{definition}
	\label{Ibor-rate}
	The \emph{forward Ibor rate} $L(t,T,\delta)$ at $t\in[0,T]$ for tenor $\delta\in\cD$ and maturity $T>0$ is the unique value $K$ satisfying $\pifra(t,T,\delta,K)=0$.
\end{definition}

Due to the affine property of FRA prices combined with the above definition, the fundamental representation 
\begin{equation*}
	\pifra(t,T,\delta,K) = \delta\bigl(L(t,T,\delta)-K\bigr)P(t,T+\delta),
\end{equation*}
follows immediately for $t\leq T$, while for $t\in[T,T+\delta]$ we have of course $$\pifra(t,T,\delta,K)=\delta(L(T,T,\delta)-K)P(t,T+\delta).$$

Starting from this expression, under no additional assumptions, we can decompose the value of the floating leg of the FRA into a  {\em multiplicative spread} and a tenor-dependent {\em discount factor}. Indeed, setting $\bar{K}(\delta):=1+\delta K$, we can write
\begin{align} 
	\pifra(t,T,\delta,K) 
	&= \bigl(1+\delta L(t,T,\delta)\bigr) P(t,T+\delta) - \bar K(\delta) P(t,T+\delta) \notag\\
	&=: S_t^\delta  P(t,T,\delta) - \bar{K}(\delta)P(t,T+\delta), \label{eq:FRA}
\end{align}
where $S^{\delta}_t$ represents a multiplicative spread and $P(t,T,\delta)$ a discount factor satisfying $P(T,T,\delta)=1$, for all $T\geq0$ and $\delta\in\cD$. 
More precisely, it holds that
\begin{align*}
	S_t^\delta  
	= P(t,t+\delta)\bigl(1+\delta L(t,t,\delta)\bigr)
	= \frac{1+\delta L(t,t,\delta)}{1+\delta F(t,t,\delta)},
\end{align*}
where $F(t,t,\delta)$ denotes the simply compounded OIS rate at date $t$ for the period $[t,t+\delta]$. 
The discount factor $P(t,T,\delta)$ is therefore given by
\[
P(t,T,\delta) = \frac{P(t,T+\delta)}{P(t,t+\delta)}\frac{1+\delta L(t,T,\delta)}{1+\delta L(t,t,\delta)}. \]
We shall sometimes refer to $P(\cdot,T,\delta)$ as {\em $\delta$-tenor bonds}. These bonds essentially span the term structure, while $S^\delta $ accounts for the counterparty and liquidity risks in the interbank market, which do not vanish as $t \to T$.

\begin{remark}	\label{rem:pre_crisis}
	In the classical pre-crisis single curve setup, the FRA price is  given by the textbook formula
	\[
	\pifra(t, T, \delta,  K) = P(t, T) - P(t, T+\delta)\bar{K}(\delta).
	\] 
	The single curve setting can be recovered from our approach by  setting $S^\delta \equiv1$ and $P(t,T,\delta):=P(t,T)$, for all $\delta\in\cD$ and $0\leq t\leq T<+\infty$. 
	This also highlights that, in a single curve setup, FRA prices are fully determined by OIS bond prices. 
\end{remark}

\begin{remark}
	\label{rem:FX_analogy}
	Representation \eqref{eq:FRA} allows for a natural interpretation via a {\em foreign exchange analogy}, following some ideas going back to \cite{Bianchetti10}. Indeed, Ibor rates can be thought of as simply compounded rates in a foreign economy, with the currency risk playing the role of the counterparty and liquidity risks of interbank transactions. 
	In this perspective, $P(t,T,\delta)$ represents the price at date $t$ (in units of the foreign currency) of a foreign zero-coupon bond with maturity $T$, while $S^{\delta}_t$ represents the spot exchange rate between the foreign and the domestic currencies. 
	The quantity $S^{\delta}_tP(t,T,\delta)$ appearing in \eqref{eq:FRA} corresponds to the value at date $t$ (in units of the domestic currency) of a payment of one unit of the foreign currency at maturity $T$. 
	In view of Remark \ref{rem:pre_crisis}, the pre-crisis scenario assumes the absence of currency risk, in which case  $S^{\delta}_tP(t,T,\delta)=P(t,T)$.
	Related foreign exchange interpretations of multiplicative spreads have been discussed in \cite{CuchieroFontanaGnoatto2016}, \cite{Macrina2018} and \cite{NguyenSeifried2015}. 
\end{remark}

With the additional assumption that OIS and $\delta$-tenor bond prices are of HJM form, we obtain our second fundamental representation \eqref{eq:FRAintro}. In the following, we will show that such a representation allows for a precise characterization of arbitrage-free multiple curve markets and leads to interesting specifications by means of  affine semimartingales. 	

\section{An extended HJM approach to term structure modeling}\label{sec:TSM}

In this section, we present a general framework for modeling the term structures of OIS bonds and FRA contracts, 
inspired by the seminal work by \cite{HJM}. 
We work in an infinite time horizon (models with a finite time horizon $T<+\infty$ can be treated by stopping the relevant processes at $T$).
As mentioned in the introduction, a key feature of the proposed framework is that we allow for the presence of stochastic discontinuities, occurring in correspondence of a countable set of predetermined dates $(T_n)_{n\in\N}$, with $T_{n+1}>T_n$, for every $n\in\N$, and $\lim_{n\rightarrow+\infty}T_n=+\infty$.

We assume that the stochastic basis $(\Omega,\ccF,\bbF,\Q)$ supports a $d$-dimensional Brownian motion $W = (W_t)_{t\geq0}$ together with an integer-valued random measure $\mu(dt,dx)$ on $\R_+\times E$, with compensator $\nu(dt,dx)=\lambda_t(dx)dt$, where $\lambda_t(dx)$ is a kernel from $(\Omega\times\R_+,\mathcal{P})$ into $(E,\cB_E)$, with $\mathcal{P}$ denoting the predictable sigma-field on $\Omega\times\R_+$ and $(E,\cB_E)$ a Polish space with its Borel sigma-field.
We refer to \cite{JacodShiryaev} for all unexplained notions related to stochastic calculus.

As a first ingredient, we assume the existence of a general {\em  num\'eraire} process $X^0 = (X^0_t)_{t\geq0}$, given by a strictly positive semimartingale admitting the representation
\beq	\label{eq:num}
X^0 = \mathcal{E}\bigl(B + H\cdot W + L\ast (\mu - \nu)\bigr),
\eeq 
where $H = (H_t)_{t\geq0}$ is an $\R^d$-valued progressively measurable process so that $\int_0^T\|{H}_s\|^2 ds < + \infty$ a.s. for all $T>0$ and $L:\Omega\times\R_+\times E\rightarrow(-1,+\infty)$ is a $\mathcal{P}\otimes\cB_E$-measurable function satisfying $\int_0^T\int_E(L^2(t,x)\wedge|L(t,x)|)\lambda_t(dx)dt<+\infty$ a.s. for all $T>0$.
Note that, in view of \cite[Theorem II.1.33]{JacodShiryaev}, the last condition is necessary and sufficient for the well-posedness of the stochastic integral $L\ast(\mu-\nu)$.
The process $B=(B_t)_{t\geq0}$ is assumed to be a finite variation process of the form
\beq  \label{eq:num_FV}
B_t = \int_0^tr_sds + \sum_{n\in\N}\Delta B_{T_n}\ind{T_n \leq t}, 
\qquad\text{ for all }t\geq0,
\eeq
where $r = (r_t)_{t\geq0}$ is an adapted process satisfying $\int_0^T|r_s|ds< + \infty$ a.s. for all $T>0$ and $\Delta B_{T_n}$ is an $\ccF_{T_n}$-measurable random variable taking values in $(-1,+\infty)$, for each $n\in\N$. Note that this specification of $X^0$ explicitly allows for jumps at times $(T_n)_{n \in \N}$, the stochastic discontinuity points of $X^0$.
The assumption that $\lim_{n\rightarrow+\infty}T_n=+\infty$ ensures that the summation in \eqref{eq:num_FV} involves only a finite number of terms, for every $t\geq0$.

\begin{remark}	\label{rem:setting_generality}
	Requiring minimal assumptions on $X^0$ enables us to unify  different modeling approaches. Usually, it is postulated that $X^0 = \exp(\int_0^{\cdot}\rois_sds)$, with $\rois$ representing the OIS short rate.
	In the setting considered here, $X^0$ can also be generated by a sequence of OIS bonds rolled over at dates $(T_n)_{n\in\N}$, compare \cite[Definition 5]{KleinSchmidtTeichmann2016} for a precise notion. This allows to avoid the unnecessary assumption of  existence of a bank account.
	In market models, the usual choice for $X^0$ is the OIS-bond with the longest available maturity, see Remark \ref{rem:termBond}.
	Moreover, it is also possible to choose $\Q$ as the physical probability measure and $X^0$ as the growth-optimal portfolio. By this, we cover the  {\em benchmark approach} to term structure modeling (see \cite{BrutiLiberati_et_al2010} and \cite{PlatenHeath}).
		While these examples refer to situations where the num\'eraire $X^0$ is tradable, {\em we do not necessarily assume that $X^0$ represents the price process of a traded asset or portfolio} (with the exception of Section \ref{sec:NAFLVR}). This generality yields additional flexibility, since $X^0$ may also represent a state-price density or pricing kernel in the spirit of \cite{Con92}, embedding a choice of the discounting asset and a probability change into a single process (compare also with Remark \ref{rem:deflators}). As explained below, the focus of Sections \ref{sec:TSM}--\ref{sec:affine} will be on deriving necessary and sufficient conditions for the local martingale property of $X^0$-denominated prices under $\Q$.
\end{remark}

The reference probability measure $\Q$ is said to be a {\em risk-neutral measure} for the multiple curve financial market with respect to $X^0$ if the $X^0$-denominated price process of every asset included in Definition \ref{def:market} is a $\Q$-local martingale.
One of our main goals consists in deriving necessary and sufficient conditions for $\Q$ to be a risk-neutral measure.
In Section \ref{sec:NAFLVR}, under the additional assumption that the num\'eraire $X^0$ is tradable, we will prove a fundamental theorem characterizing absence of arbitrage in the sense of NAFLVR, for which the existence of a risk-neutral measure is a sufficient condition (see Remark \ref{rem:ELMM}). 

In view of representation \eqref{eq:FRA}, modeling a multiple curve financial market requires the specification of multiplicative spreads $S^{\delta}$ and $\delta$-tenor bond prices, for $\delta\in\cD$.
The  {\em multiplicative spread} process $S^\delta  = (S_t^\delta )_{t\geq0}$ is assumed to be a strictly positive semimartingale, for each $\delta\in\cD$. Similarly as in \eqref{eq:num}, we assume that  $S^\delta $ admits the representation
\beq  \label{eq:spread}
S^\delta 
= S_0^\delta \, \mathcal{E}\bigl(A^{\delta} + H^{\delta}\cdot W 
+ L^{\delta} \ast (\mu - \nu)\bigr),
\eeq
for every $\delta\in\cD$, where $A^{\delta}$, $H^{\delta}$ and $L^{\delta}$ satisfy the same requirements of the processes $B$, $H$ and $L$, respectively, appearing in \eqref{eq:num}.
In line with \eqref{eq:num_FV}, we furthermore assume that
\beq \label{eq:spread_FV}
A_t^{\delta} = \int_0^t\alpha_s^{\delta}ds + \sum_{n\in\N}\Delta A^{\delta}_{T_n}\ind{T_n \leq t}, 
\qquad\text{ for all }t\geq0,
\eeq
where $(\alpha_t^{\delta})_{t\geq0}$ is an adapted  process satisfying $\int_0^T|\alpha_s^{\delta}|ds< + \infty$ a.s., for all $\delta\in\cD$ and $T>0$, and $\Delta A_{T_n}^{\delta}$ is an $\ccF_{T_n}$-measurable random variable taking values in $(-1,+\infty)$, for each $n\in\N$ and $\delta\in\cD$.

We let $P(t, T, 0): = P(t, T)$, for all $0 \leq t \leq T<+\infty$.
We assume that, for every $T\in\R_+$ and $\delta\in\cD_0:=\cD\medcup\{0\}$, the $\delta$-tenor bond price process $(P(t, T, \delta))_{0 \leq t \leq T}$ is of the form
\beq	\label{eq:PtT}
P(t, T, \delta) = \exp\left( - \int_{(t, T]} f(t, u, \delta) \eta(du)\right), 
\qquad\text{ for all }0\leq t\leq T,
\eeq
where
\beq	\label{eq:eta}
\eta(du) = du + \sum_{n\in\N}\delta_{T_n}(du).
\eeq
Note that $\eta([0,T])<+\infty$, for all $T>0$.
We adopt the convention $\int_{(T, T]} f(T, u, \delta) \eta(du) = 0$, for all $T\in\R_+$ and $\delta\in\cD_0$.
For every  $T\in\R_+$ and $\delta\in\cD_0$, we assume that the {\em forward rate} process $(f(t,T,\delta))_{0\leq t\leq T}$  satisfies
\begin{align}\label{eq:fwd_rate}
	f(t,T,\delta) &=  f(0,T,\delta) + \int_0^t a(s,T,\delta) ds 
	+ V(t, T, \delta)
	+ \int_0^t b(s,T,\delta) dW_s \nonumber\\
	&\quad + \int_0^t  \int_E g(s, x, T,\delta) \bigl(\mu(ds,dx) - \nu(ds,dx)\bigr),
\end{align}
for all $0\leq t\leq T$, where $V(\cdot, T, \delta) = V(t, T, \delta)_{0 \leq t \leq T}$ is a pure jump adapted process of the form
\[
V(t, T, \delta) = \sum_{n\in\N} \Delta V(T_n,T,\delta) \ind{T_n \leq t},
\qquad\text{ for all }0\leq t\leq T,
\]
with $\Delta V(t,T,\delta) = 0$ for all $0 \leq T < t <+\infty$.
Moreover, for all $n\in\N$, $T\in\R_+$ and $\delta\in\cD_0$, we also assume that $\int_0^T|\Delta V(T_n,u,\delta)|du<+\infty$.

\begin{remark}\label{rem:disc_HJM}
	\begin{enumerate}
		\item The above framework allows for a general modeling of type I and type II stochastic discontinuities (see Section \ref{sec:discontinuities}), as we illustrate by means of explicit examples in Section \ref{sec:affine}. Moreover, the dependence on $\delta$ in equations \eqref{eq:spread}--\eqref{eq:fwd_rate} allows the discontinuities to have a different impact on different yield curves. This is consistent with the typical market behavior, which shows a dampening of the discontinuities over longer tenors.
		\item The discontinuity dates $(T_n)_{n\in\N}$ play two distinct but equally important roles. On the one hand, they represent stochastic discontinuities in the dynamics of all relevant processes. On the other hand, they represent discontinuity points in maturity of bond prices (see equation \eqref{eq:PtT}). As shown in Theorem \ref{thm:HJM} below, absence of arbitrage will imply a precise relation between these two aspects.
	\end{enumerate}
\end{remark}

\begin{assumption}\label{ass}
	The following conditions hold a.s. for every $\delta\in\cD_0$:
	\begin{enumerate}[(i)]
		\item 
		the {\em initial forward curve} $T\mapsto f(0, T, \delta)$ is $\cF_0\otimes\cB(\R_+)$-measurable, real-valued and satisfies $\int_0^T|f(0, u, \delta)|du < + \infty$, for all $T\in\R_+$;
		\item
		the {\em drift process} $a(\cdot, \cdot,\delta): \Omega \times\R_+^2\rightarrow\R$ is a real-valued, progressively measurable process, in the sense that the restriction
		$$ a(\cdot, \cdot, \delta)|_{[0,t]}:\Omega\times[0,t]\times\R_+\rightarrow\R$$ is $\mathcal{F}_t\otimes \mathcal{B}([0,t])\otimes \mathcal{B}(\R_+)$-measurable, for every $t\in\R_+$. Moreover, it satisfies for all $0 \leq T < t <+\infty$ that $a(t,T,\delta) = 0$, and
		\begin{equation*}
			\int_0^T\int_0^u |a(s,u,\delta)|ds\,\eta(du) < + \infty,
			\qquad\text{ for all }T>0;
		\end{equation*}
		\item 
		the {\em volatility process} $b(\cdot, \cdot, \delta):\Omega\times\R_+^2\rightarrow\R^d$ is an $\R^d$-valued progressively measurable process, in the sense that the restriction 
		$$ b(\cdot, \cdot, \delta)|_{[0,t]}:\Omega\times[0,t]\times\R_+\rightarrow\R^d $$is $\mathcal{F}_t\otimes \mathcal{B}([0,t])\otimes \mathcal{B}(\R_+)$-measurable, for every $t\in\R_+$. Moreover, it satisfies for all $0 \leq T < t <+\infty$ that $b(t,T,\delta) = 0$, and
		\[
		\sum_{i=1}^d\int_0^T\left(\int_0^u|b^i(s, u, \delta)|^2ds\right)^{1/2}\eta(du) < + \infty,
		\qquad\text{ for all }T>0;
		\]
		\item the {\em jump function} $g(\cdot, \cdot, \cdot, \delta):\Omega\times\R_+\times E\times\R_+\rightarrow\R$ is a $\mathcal{P}\otimes \cB_E\otimes\mathcal{B}(\R_+)$-measurable real-valued function satisfying $g(t, x, T, \delta) = 0$ for all $x \in E$ and $0 \leq T < t <+\infty$. Moreover, it satisfies
		\begin{align*}
			\int_{0}^{T}\int_E  \int_0^T |g(s,x,u,\delta)|^2 \eta(du)\nu(ds,dx) < +\infty,
			\qquad\text{ for all }T>0.    	
		\end{align*}
	\end{enumerate}
\end{assumption}

Assumption \ref{ass} implies that the integrals appearing in the forward rate equation \eqref{eq:fwd_rate} are well-defined for $\eta$-a.e. $T\in\R_+$. Moreover, the integrability requirements appearing in conditions (ii)--(iv) of Assumption \ref{ass} ensure that we can apply ordinary and stochastic Fubini theorems, in the versions of \cite{Veraar12} for the Brownian motion and Proposition A.2 in \cite{Bjoerk1997} for the compensated random measure.
The mild measurability requirement in conditions (ii)--(iii) holds if $a(\cdot,\cdot,\delta)$ and $b(\cdot,\cdot,\delta)$ are $\mathcal{P}_{rog}\otimes\mathcal{B}(\R_+)$-measurable, for every $\delta\in\cD_0$, with $\mathcal{P}_{rog}$ denoting the progressive sigma-algebra on $\Omega\times\R_+$, see \cite[Remark 2.1]{Veraar12}.

\begin{remark}	\label{rem:measure}
	There is no loss of generality in taking a single measure $\eta$ instead of different measures $\eta^\delta$ for each tenor $\delta\in\cD_0$. Indeed, dependence on the tenor can be embedded in our framework by suitably specifying each forward rate $f(t,T,\delta)$ in \eqref{eq:fwd_rate} and using a common measure $\eta = \sum_{\delta \in \cD_0}\eta^\delta$.
\end{remark}

For all $0 \leq t \leq T <+\infty$, $\delta\in\cD_0$ and $x \in E$, we set
\begin{align*}
	\bar{a}(t, T, \delta) &: = \int_{[t, T]} a(t, u, \delta) \eta(du),\\	
	\bar{b}(t, T, \delta) &: = \int_{[t, T]} b(t, u, \delta) \eta(du), \\
	\bar{V}(t, T, \delta) &: = \int_{[t, T]} \Delta V(t, u, \delta) \eta(du),\\
	\bar{g}(t, x, T, \delta) &: = \int_{[t, T]} g(t, x, u, \delta) \eta(du).
\end{align*}

As a first result, the following lemma (whose proof is postponed to Appendix \ref{app:proofs}) gives a semimartingale representation of the process $P(\cdot, T, \delta)$.

\begin{lemma}	\label{lem:bond_prices}
	Suppose that Assumption \ref{ass} holds. Then, for every $T\in\R_+$ and $\delta\in\cD_0$, the process $(P(t, T, \delta))_{0 \leq t \leq T}$ is a semimartingale and admits the representation
	\begin{align}  \label{eq:dec_PtT}
		P(t, T, \delta) = \exp\bigg(&-\int_0^T f(0, u, \delta)\eta(du) 
		- \int_0^t \bar{a}(s, T, \delta) ds  - \sum_{n\in\N} \bar{V}(T_n, T, \delta)\ind{T_n \leq t}\nonumber\\
		&
		- \int_0^t \bar{b}(s, T, \delta) dW_s  - \int_0^t \int_E \bar{g}(s, x, T,\delta) \bigl(\mu(ds,dx) - \nu(ds,dx)\bigr)\nonumber \\
		& + \int_0^t f(u, u, \delta)\eta(du) \bigg),
		\qquad\qquad\text{ for all }0\leq t\leq T.
	\end{align}                         
\end{lemma}

The $\delta$-tenor bond price process $(P(t, T, \delta))_{0\leq t\leq T}$ admits an equivalent representation as a stochastic exponential, which will be used in the following. The following corollary is a direct consequence of Lemma \ref{lem:bond_prices} and \cite[Theorem II.8.10]{JacodShiryaev}, using the fact that $\mu(\{T_n\}\times E)=0$ a.s., for all $n\in\N$.

\begin{corollary} 	\label{cor:bond_stoch_exp}
	Suppose that Assumption \ref{ass} holds. Then, for every $T\in\R_+$ and $\delta\in\cD_0$, the process $P(\cdot,T,\delta)=(P(t, T, \delta))_{0 \leq t \leq T}$ admits the representation
	\begin{align*} 
		\begin{aligned}
			P(\cdot, T, \delta) = \mathcal{E} \bigg(
			& - \int_0^T f(0, u, \delta)\eta(du) 
			- \int_0^\cdot \bar{a}(s, T, \delta) ds 
			+ \frac{1}{2} \int_0^\cdot \|\bar{b}(s, T, \delta)\|^2 ds \\
			& - \int_0^\cdot \bar{b}(s, T, \delta) dW_s 
			- \int_0^\cdot \int_E \bar{g}(s, x, T,\delta) \bigl(\mu(ds,dx) - \nu(ds,dx)\bigr) \\
			& + \int_0^\cdot \int_E \bigl(e^{-\bar{g}(s, x, T, \delta)}-1 + \bar{g}(s, x, T, \delta)\bigr) \mu(ds,dx)  \\
			& + \int_0^\cdot f(u, u, \delta)du +\sum_{n\in\N} \bigl(e^{-\bar{V}(T_n, T, \delta) + f(T_n,T_n,\delta)} -1 \bigl)\Ind_{\dbraco{T_n,+\infty}} \bigg).
		\end{aligned}
	\end{align*}
\end{corollary}

We are now in a position to state the central result of this section, which provides necessary and sufficient conditions for the reference probability measure $\Q$ to be a risk-neutral measure with respect to the num\'eraire $X^0$.
We recall that a random variable $\xi$ on $(\Omega,\ccF,\Q)$ is said to be {\em sigma-integrable} with respect to a sigma-field $\ccG\subseteq\ccF$ if there exists a sequence of measurable sets $(\Omega_n)_{n\in\N}\subseteq\ccG$ increasing to $\Omega$ such that $\xi\,\Ind_{\Omega_n}\in L^1(\Q)$ for every $n\in\N$, see Definition 1.15 in \cite{he1992semimartingale}. 
A random variable $\xi$ is sigma-finite with respect to $\ccG$ if and only if the generalized conditional expectation $\E^{\Q}[\xi|\ccG]$ is a.s. finite.
For convenience of notation, let $\alpha_t^0 := 0$, $H_t^0:=0$, $L^0(t,x):= 0$ and $\Delta A_{T_n}^0 := 0$ for all $n\in\N$, $t\in\R_+$ and $x\in E$, so that $S^0: = \mathcal{E}(A^0 + H^0 \cdot W + L^0\ast (\mu - \nu))\equiv1$. 
Let 
\begin{align*}
	\Lambda(s,x,T,\delta) := \frac{1+L^{\delta}(s,x)}{1+L(s,x)}e^{-\bar{g}(s,x,T,\delta)}+L(s,x)-L^{\delta}(s,x)+\bar{g}(s,x,T,\delta)-1.
\end{align*}

\begin{theorem} \label{thm:HJM}
	Suppose that Assumption \ref{ass} holds. Then $\Q$ is a risk-neutral measure with respect to $X^0$ if and only if, for every $\delta\in\cD_0$,
	\begin{align}	\label{eq:jumps_int}
		\int_0^T\int_E \big| \Lambda(s,x,T,\delta)\big| \, \lambda_s(dx)ds < +\infty
	\end{align}
	a.s. for every $T\in\R_+$ and, for every $n\in\N$ and $T\geq T_n$, the random variable
	\beq	\label{eq:jump_sigma_int}
	\frac{1+\Delta A_{T_n}^{\delta}}{1+\Delta B_{T_n}}e^{-\int_{(T_n,T]}\Delta V(T_n,u,\delta)\eta(du)}
	\eeq
	is sigma-integrable with respect to $\ccF_{T_n-}$, and the following four conditions hold a.s.:
	\begin{enumerate}
		\item[(i)] 
		for a.e. $t\in\R_+$, it holds that
		\[
		r_t - \alpha_t^{\delta} =  f(t,t,\delta) -  H_t^{\top}H_t^{\delta} + \|H_t\|^2 
		+ \int_E\frac{L(t,x)}{1+L(t,x)}\bigl(L(t,x)-L^{\delta}(t,x)\bigr)\lambda_t(dx);
		\]
		\item[(ii)]
		for every $T\in\R_+$ and for a.e. $t\in[0,T]$, it holds that
		\begin{align*}
			\bar{a}(t,T,\delta)
			&= \frac{1}{2} \|\bar{b}(t,T,\delta)\|^2 
			+ \bar{b}(t,T,\delta)^{\top}\bigl(H_t - H_t^{\delta}\bigr) \\
			&\quad + \int_E\left(\frac{1+L^{\delta}(t,x)}{1+L(t,x)}\bigl(e^{-\bar{g}(t,x,T,\delta)}-1\bigr)+\bar{g}(t,x,T,\delta)\right)\lambda_t(dx);
		\end{align*}
		\item[(iii)]
		for every $n\in\N$, it holds that
		\[
		\E^{\Q}\left[\frac{1+\Delta A_{T_n}^{\delta}}{1+\Delta B_{T_n}}\Bigg|\ccF_{T_n-}\right] = e^{-f(T_n-,T_n,\delta)};
		\]
		\item[(iv)]
		for every $n\in\N$ and $T\geq T_n$, it holds that
		\[
		\E^{\Q}\left[\frac{1+\Delta A_{T_n}^{\delta}}{1+\Delta B_{T_n}}\left(e^{-\int_{(T_n,T]}\Delta V(T_n,u,\delta)\eta(du)}-1\right)\Bigg|\ccF_{T_n-}\right] = 0.
		\]
	\end{enumerate}
\end{theorem}

\begin{remark}  \label{rem:condition1}
	By considering separately the cases $\delta=0$ and $\delta\in\cD$, we can obtain a more explicit version of condition {\em (i)} of Theorem \ref{thm:HJM}, which is equivalent to the validity of the two conditions, for every $\delta\in\cD$ and a.e. $t\in\R_+$,
	\begin{align}
		r_t &= f(t,t,0) + \|H_t\|^2 + \int_E \frac{L^2(t,x)}{1+L(t,x)}\lambda_t(dx);  \label{eq:HJM_r}\\
		\alpha_t^{\delta} &= f(t,t,0) - f(t,t,\delta)
		+  H_t^{\top}H_t^{\delta}  
		+ \int_E \frac{L^{\delta}(t,x)L(t,x)}{1+L(t,x)}\lambda_t(dx).
		\label{eq:HJM_s}
	\end{align}
\end{remark}

The conditions of Theorem \ref{thm:HJM} together with Remark \ref{rem:condition1}  admit the following  interpretation.
First, for $\delta=0$ condition {\em (i)} requires that the drift rate $r_t$ of the num\'eraire process $X^0$ equals the short end of the instantaneous yield $f(t,t,0)$ on OIS bonds, plus two additional terms accounting for the volatility of $X^0$ itself.\footnote{Note that, at the present level of generality, the rate $r_t$ does not represent a riskless rate of return.} For $\delta\neq0$, condition {\em (i)} requires that, at the short end, the instantaneous yield $\alpha^{\delta}_t+f(t,t,\delta)$ on the floating leg of a FRA equals the instantaneous yield $f(t,t,0)$ plus a risk premium determined by the covariation between the num\'eraire process $X^0$ and the multiplicative spread process $S^{\delta}$.

Second, condition {\em (ii)} is a generalization of the well-known HJM drift condition. In particular, if $\cD=\emptyset$ and the process $X^0$ does not have local martingale components, then condition {\em (ii)} reduces to the  drift restriction established in Proposition 5.3 of \cite{Bjoerk1997} for single-curve jump-diffusion models.

Finally, conditions {\em (iii)} and {\em (iv)} are new and specific to our setting with stochastic discontinuities. Together, they correspond to excluding the possibility that, at some predetermined date $T_n$, prices of $X^0$-denominated assets  exhibit jumps whose size can be predicted on the basis of the information contained in $\cF_{T_n-}$.
Indeed, such a possibility would violate absence of arbitrage (compare with \cite{FPP18}).

\begin{proof}\textit{of Theorem \ref{thm:HJM} } 
	Recall that $P(t, T, 0) = P(t, T)$, $0 \leq t \leq T<+\infty$.
	By definition, $\Q$ is a risk-neutral measure with respect to $X^0$ if and only if the processes $P(\cdot, T)/X^0$ and $\pifra(\cdot, T, \delta, K)/X^0$ are $\Q$-local martingales, for every $T\in\R_+$, $\delta\in\cD$ and $K\in\R$.
	In view of \eqref{eq:FRA} and using the notational convention introduced above, this holds if and only if the process $S^\delta P(\cdot, T, \delta)/X^0$ is a $\Q$-local martingale, for every $T\in\R_+$ and $\delta\in\cD_0$.
	An application of Corollary \ref{cor:stoch_exp_property} together with Corollary \ref{cor:bond_stoch_exp} and equations \eqref{eq:num}--\eqref{eq:spread_FV} yields
	\begin{align}	\label{eq:prepdynSPX}
		\frac{S^\delta P(\cdot,T,\delta)}{X^0}
		= & S_0^\delta P(0,T,\delta)\nonumber\\
		&\cdot
		\mathcal{E}\left(\int_0^{\cdot}k_s(T,\delta)ds  + K^{(1)}(T,\delta) + K^{(2)}(T,\delta) + M(T,\delta)\right),
	\end{align}
	where $(k_t(T,\delta))_{0\leq t\leq T}$ is an adapted process given by
	\[
	k_t(T,\delta) 
	:=
	\alpha_t^{\delta} - r_t -\bar{a}(t,T,\delta) + \frac{1}{2} \|\bar{b}(t,T,\delta)\|^2 
	+ f(t,t,\delta) 
	 + \bar{b}(t,T,\delta)^{\top}\bigl(H_t - H_t^{\delta}\bigr) 
	-  H_t^{\top}H_t^{\delta} + \|H_t\|^2,
	\]
	$(K^{(1)}_t(T,\delta))_{0\leq t\leq T}$ is a pure jump finite variation  process given by
	\begin{align*}
		K^{(1)}_t(T,\delta) &:=
		\int_0^t \int_E \left(e^{-\bar{g}(s, x, T, \delta)}-1 + \bar{g}(s, x, T, \delta)\right) \mu(ds,dx) 	\\
		&\quad + \int_0^\cdot \int_E \frac{L(s, x)}{1+L(s, x)} \left(-L^{\delta}(s, x)  - \bigl(e^{-\bar{g}(s, x, T, \delta)} - 1\bigr) + L(s, x) \right) \mu(ds,dx)\\
		&\quad + \int_0^\cdot \int_E \frac{ L^{\delta}(s, x)}{1+L(s, x)} \left(e^{-\bar{g}(s, x, T, \delta)} - 1\right) \mu(ds,dx)	\\
		&\;= \int_0^t\int_E \Lambda(s,x,T,\delta) \mu(ds,dx),
	\end{align*}
	and $(K^{(2)}_t(T,\delta))_{0\leq t\leq T}$ is a pure jump finite variation  process given by
	\[	\begin{aligned}
	K^{(2)}_t(T,\delta)
	&:= \sum_{n\in\N} \ind{T_n\leq t}
	\left( \frac{\Delta A_{T_n}^{\delta} }{1 + \Delta B_{T_n}} 
	+ \frac{1}{1 + \Delta B_{T_n}}\left(e^{-\bar{V}(T_n, T, \delta) + f(T_n, T_n, \delta)} - 1\right)\right.\\
	&\quad \left.- \frac{\Delta B_{T_n}}{1 + \Delta B_{T_n}}
	+ \frac{\Delta A_{T_n}^{\delta} }{1 + \Delta B_{T_n}}
	\left(e^{-\bar{V}(T_n, T, \delta) + f(T_n, T_n, \delta)} - 1\right) \right)	\\
	&= \sum_{n\in\N} \ind{T_n\leq t}
	\left(\frac{1+\Delta A_{T_n}^{\delta}}{1+\Delta B_{T_n}}e^{-\int_{(T_n,T]}\Delta V(T_n,u,\delta)\eta(du) + f(T_n-,T_n,\delta)}-1\right),
	\end{aligned}	\]
	where in the last equality we made use of \eqref{eq:fwd_rate} together with the definition of the process $\bar{V}$. 
	The process $M(T,\delta)=(M_t(T,\delta))_{0\leq t\leq T}$ appearing in \eqref{eq:prepdynSPX} is the local martingale
	\[	\begin{aligned}
	M_t(T,\delta) &:= \int_0^t \bigl( H_s^{\delta} - H_s -\bar{b}(s, T, \delta)  \bigr) dW_s \\
	&\quad + \int_0^t \int_E \bigl(L^{\delta}(s, x)- L(s,x) -\bar{g}(s, x, T,\delta)\bigr) \bigl(\mu(ds,dx) - \nu(ds,dx)\bigr).
	\end{aligned}	\]
	Note that the set $\{\Delta K^{(1)}(T,\delta)\neq0\}\medcap\{\Delta K^{(2)}(T,\delta)\neq0\}$ is evanescent for every $T\in\R_+$ and $\delta\in\cD_0$, as a consequence of the fact that $\mu(\{T_n\}\times E)=0$ a.s. for all $n\in\N$.
	
	Suppose that $S^\delta P(\cdot,T,\delta)/X^0$ is a $\Q$-local martingale, for every $T\in\R_+$ and $\delta\in\cD_0$. 
	In this case, \eqref{eq:prepdynSPX} implies that the finite variation process $\int_0^{\cdot}k_s(T,\delta)ds+K^{(1)}(T,\delta)+K^{(2)}(T,\delta)$ is also a $\Q$-local martingale. By means of \cite[Lemma I.3.11]{JacodShiryaev}, this implies that the pure jump finite variation process $K^{(1)}(T,\delta)+K^{(2)}(T,\delta)$ is of locally integrable variation.
	Since the two processes $K^{(1)}(T,\delta)$ and $K^{(2)}(T,\delta)$ do not have common jumps, it holds that 
	$$|\Delta K^{(i)}(T,\delta)|\leq|\Delta K^{(1)}(T,\delta)+\Delta K^{(2)}(T,\delta)|,  \quad \text{ for } i=1,2.
	$$
	As a consequence of this fact, both processes $K^{(1)}(T,\delta)$ and $K^{(2)}(T,\delta)$ are of locally integrable variation.
	Noting that 
	$$ 
	K^{(2)}(T,\delta)=\sum_{n\in\N}\Delta K_{T_n}^{(2)}(T,\delta)\Ind_{\dbraco{T_n,+\infty}},
	$$
	Theorem 5.29 of \cite{he1992semimartingale} implies that the random variable $\Delta K^{(2)}_{T_n}(T,\delta)$ is sigma-integrable with respect to $\ccF_{T_n-}$, for every $n\in\N$. This is equivalent to the sigma-integrability  of
	\beq	\label{eq:jump_at_Tn}
	\frac{1+\Delta A_{T_n}^{\delta}}{1+\Delta B_{T_n}}e^{-\int_{(T_n,T]}\Delta V(T_n,u,\delta)\eta(du) + f(T_n-,T_n,\delta)}
	\eeq
	with respect to $\ccF_{T_n-}$.
	Since $f(T_n-,T_n,\delta)$ is $\ccF_{T_n-}$-measurable, the sigma-integrability of \eqref{eq:jump_at_Tn} with respect to $\ccF_{T_n-}$ can be equivalently stated as the sigma-integrability of \eqref{eq:jump_sigma_int} with respect to $\ccF_{T_n-}$.
	Moreover, the fact that $K^{(1)}(T,\delta)$ is of locally integrable variation is equivalent to the a.s. finiteness of the integral
	\[
	\int_0^T\int_E \bigl|\Lambda(s,x,T,\delta) \bigr|\, \lambda_s(dx)ds,
	\]
	thus proving the integrability conditions \eqref{eq:jumps_int}, \eqref{eq:jump_sigma_int}.
	Having established that the two processes $K^{(1)}(T,\delta)$ and $K^{(2)}(T,\delta)$ are of locally integrable variation, we can take their compensators (dual predictable projections), see \cite[Theorem I.3.18]{JacodShiryaev}. 
	This leads to 
	\beq	\label{eq:SPX_pred}
	\begin{aligned}
		\frac{S^\delta P(\cdot,T,\delta)}{X^0}
		= S_0^\delta P(0,T,\delta)\,\mathcal{E}\left(\int_0^{\cdot}\hat{k}_s(T,\delta)ds+\widehat{K}^{(2)}(T,\delta) + M'(T,\delta)\right),
	\end{aligned}	
	\eeq
	where 
	\begin{equation} \label{eq:comp_0}
		M'(T,\delta):= M(T,\delta)+K^{(1)}(T,\delta)+K^{(2)}(T,\delta)
		-\int_0^{\cdot}\bigl(\hat{k}_s(T,\delta)-k_s(T,\delta)\bigr)ds
		-\widehat{K}^{(2)}(T,\delta)
	\end{equation}  
	is a local martingale, $(\hat{k}_t(T,\delta))_{0\leq t\leq T}$ is an adapted process given by
	\begin{equation*}
		\hat{k}_t(T,\delta) = k_t(T,\delta) 
		+ \int_E \Lambda(t,x,T,\delta) \, \lambda_t(dx)
	\end{equation*}
	and, in view of \cite[Theorem 5.29]{he1992semimartingale},  $\widehat{K}^{(2)}(T,\delta)=\sum_{n\in\N}\Delta\widehat{K}^{(2)}_{T_n}(T,\delta)\Ind_{\dbraco{T_n,+\infty}}$ is a pure jump finite variation predictable process with
	\begin{align*}
		\Delta\widehat{K}^{(2)}_{T_n}(T,\delta)
		= e^{f(T_n-,T_n,\delta)}\E^{\Q}\left[\frac{1+\Delta A_{T_n}^{\delta}}{1+\Delta B_{T_n}}e^{-\int_{(T_n,T]}\Delta V(T_n,u,\delta)\eta(du)} \Bigg|\ccF_{T_n-}\right]-1,
	\end{align*}
	for all $n\in\N$. If $S^\delta P(\cdot,T,\delta)/X^0$ is a $\Q$-local martingale, then by equation \eqref{eq:SPX_pred} the process $\int_0^{\cdot}\hat{k}_s(T,\delta)ds+\widehat{K}^{(2)}(T,\delta)$ must be null (up to an evanescent set), being a predictable local martingale of finite variation, see \cite[Corollary I.3.16]{JacodShiryaev}. In particular, analyzing separately its absolutely continuous and discontinuous parts, this holds if and only if $\hat{k}_t(T,\delta)=0$ outside of a set of $(\Q\otimes dt)$-measure zero and $\Delta\widehat{K}^{(2)}_{T_n}(T,\delta)=0$ a.s. for every $n\in\N$.
	Let us first consider the absolutely continuous part
	\begin{align*}
		0 & = \hat{k}_t(T,\delta) \\
		&= \alpha_t^{\delta} - r_t -\bar{a}(t,T,\delta) + \frac{1}{2} \|\bar{b}(t,T,\delta)\|^2 + f(t,t,\delta) \\
		& \quad + \bar{b}(t,T,\delta)^{\top}\bigl(H_t - H_t^{\delta}\bigr) -  H_t^{\top}H_t^{\delta} + \|H_t\|^2 + \int_E \Lambda(t,x,T,\delta) \, \lambda_t(dx).
	\end{align*}
	The integral appearing in the last line is a.s. finite for a.e. $t\in[0,T]$ as a consequence of \eqref{eq:jumps_int}.
	Taking $T=t$ leads to the requirement
	\[
	r_t - \alpha_t^{\delta} =  f(t,t,\delta) -  H_t^{\top}H_t^{\delta} + \|H_t\|^2 
	+ \int_E \frac{L(t,x)}{1+L(t,x)}\bigl(L(t,x)-L^{\delta}(t,x)\bigr)\lambda_t(dx),
	\]
	for a.e. $t\in\R_+$, which gives condition {\em (i)} in the statement of the theorem. In turn, inserting this last condition into the equation $\hat{k}_t(T,\delta)=0$ directly leads to condition {\em (ii)}.
	Considering then the pure jump part, the condition $\Delta\widehat{K}^{(2)}_{T_n}(T,\delta)=0$ a.s., for all $n\in\N$, leads to
	\beq	\label{eq:general_jump_Tn}
	\E^{\Q}\left[\frac{1+\Delta A_{T_n}^{\delta}}{1+\Delta B_{T_n}}e^{-\int_{(T_n,T]}\Delta V(T_n,u,\delta)\eta(du)}\Bigg|\ccF_{T_n-}\right] = e^{-f(T_n-,T_n,\delta)} 
	\eeq
	a.s. for all $n\in\N$.
	Condition {\em (iii)} in the statement of the theorem is obtained by taking $T=T_n$, while condition {\em (iv)} follows by inserting condition {\em (iii)} into \eqref{eq:general_jump_Tn}.
	
	Conversely, if the integrability conditions \eqref{eq:jumps_int}, \eqref{eq:jump_sigma_int} are satisfied then the finite variation processes $K^{(1)}(T,\delta)$ and $K^{(2)}(T,\delta)$ appearing in \eqref{eq:prepdynSPX} are of locally integrable variation. One can therefore take their compensators and obtain representation \eqref{eq:SPX_pred}. It is then easy to verify that, if the four conditions {\em (i)}--{\em (iv)} hold, then the processes $\hat{k}(T,\delta)$ and $\widehat{K}^{(2)}(T,\delta)$ appearing in \eqref{eq:SPX_pred} are null, up to an evanescent set. This proves the local martingale property of $S^\delta P(\cdot,T,\delta)/X^0$, for every $T\in\R_+$ and $\delta\in\cD_0$.
\end{proof}

\begin{remark}
	\label{rem:FX_analogy2}
	The foreign exchange analogy introduced in Remark \ref{rem:FX_analogy} carries over to the conditions established in Theorem \ref{thm:HJM}. In particular, in the special case where $H_t=L(t,x)=0$, for all $(t,x)\in\R_+\times E$, it can be easily verified that conditions {\em (i)}--{\em (ii)} reduce exactly to the HJM conditions established in \cite{Koval2005} in the context of multi-currency HJM semimartingale models.
\end{remark}

\subsection{The OIS bank account as num\'eraire}

In HJM models, the num\'eraire is usually chosen as the OIS bank account $\exp(\int_0^{\cdot}\rois_sds)$, with $\rois$ denoting the OIS short rate. In this context, an application of Theorem \ref{thm:HJM} enables us to characterize all equivalent local martingale measures (ELMMs, see Section \ref{sec:NAFLVR}) with respect to the OIS bank account num\'eraire.
To this effect, let $\Q'$ be a probability measure on $(\Omega,\ccF)$ equivalent to $\Q$ and denote by $Z^\prime$ its density process, i.e., $Z'_t=d\Q'|_{\ccF_t}/d\Q|_{\ccF_t}$, for all $t\geq0$. 
We denote the expectation with respect to $\Q'$ by $\E^{\Q'}$ and  assume that
\beq	\label{eq:stoch_exp_Z}
Z' = \mathcal{E}\bigg(-\theta\cdot W - \psi\ast(\mu-\nu)-\sum_{n\in\N}Y_n\Ind_{\dbraco{T_n,+\infty}}\bigg),
\eeq
for an $\R^d$-valued progressively measurable process $\theta=(\theta_t)_{t\geq0}$ satisfying the integrability condition $\int_0^T\|\theta_s\|^2ds<+\infty$ a.s. for all $T>0$, a $\mathcal{P}\otimes\mathcal{B}_E$-measurable function $\psi:\Omega\times\R_+\times E\rightarrow(-\infty,+1)$ satisfying the integrability condition $\int_0^T\int_E(|\psi(s,x)|\wedge\psi^2(s,x))\lambda_s(dx)ds<+\infty$ a.s. for all $T>0$, and a family $(Y_n)_{n\in\N}$ of random variables taking values in $(-\infty,+1)$ such that $Y_n$ is $\ccF_{T_n}$-measurable and $\E^{\Q}[Y_n|\ccF_{T_n-}]=0$, for all $n\in\N$. Denote 
\begin{equation*}
	\Lambda^*(s,x,T,\delta) = \bigl(1-\psi(s,x)\bigr)\bigl((1+L^{\delta}(s,x))e^{-\bar{g}(s,x,T,\delta)}-1\bigr)  -L^{\delta}(s,x)+\bar{g}(s,x,T,\delta).
\end{equation*}

\begin{corollary}	\label{cor:ELMM_Q}
	Suppose that Assumption \ref{ass} holds. Let $\Q'$ be a probability measure on $(\Omega,\ccF)$ equivalent to $\Q$, with density process $Z'$ given  in \eqref{eq:stoch_exp_Z}.
	Assume furthermore that $\int_0^T\int_{\{\psi(s,x)\in[0,1]\}}\psi^2(s,x)/(1-\psi(s,x))\lambda_s(dx)ds<+\infty$ a.s. for all $T>0$.
	Then, $\Q'$ is an ELMM with respect to the num\'eraire $\exp(\int_0^{\cdot}\rois_sds)$ if and only if, for every $\delta\in\cD_0$,
	\begin{align}\label{eq:Q_int}
		\int_0^T\int_E \bigl| \Lambda^*(s,x,T,\delta)\bigr|\, \lambda_s(dx)ds< +\infty
	\end{align}
	a.s. for every $T\in\R_+$ and, for every $n\in\N$ and $T\geq T_n$, the random variable
	\[
	\bigl(1+\Delta A_{T_n}^{\delta}\bigr)
	e^{-\int_{(T_n,T]}\Delta V(T_n,u,\delta)\eta(du)}
	\]
	is sigma-integrable under $\Q'$ with respect to $\ccF_{T_n-}$, and the following conditions hold a.s.:
	\begin{enumerate}
		\item[(i)] 
		for a.e. $t\in\R_+$, it holds that
		\begin{align*}
			\rois_t &= f(t,t,0),	\\
			\alpha_t^{\delta} &=  f(t,t,0) - f(t,t,\delta) +  \theta_t^{\top}H_t^{\delta} + \int_E\psi(t,x)L^{\delta}(t,x)\lambda_t(dx);
		\end{align*}
		\item[(ii)]
		for every $T\in\R_+$ and for a.e. $t\in[0,T]$, it holds that
		\begin{align*}
			\bar{a}(t,T,\delta)
			&= \frac{1}{2} \|\bar{b}(t,T,\delta)\|^2 
			+ \bar{b}(t,T,\delta)^{\top}\bigl(\theta_t - H_t^{\delta}\bigr) \\
			&\quad+ \int_E\left(\bigl(1-\psi(t,x)\bigr)\bigl(1+L^{\delta}(t,x)\bigr)\bigl(e^{-\bar{g}(t,x,T,\delta)}-1\bigr)+\bar{g}(t,x,T,\delta)\right)\lambda_t(dx);
		\end{align*}
		\item[(iii)]
		for every $n\in\N$, it holds that
		\[
		\E^{\Q'}\bigl[\Delta A_{T_n}^{\delta}\bigr|\ccF_{T_n-}\bigr] = e^{-f(T_n-,T_n,\delta)} -1;
		\]
		\item[(iv)]
		for every $n\in\N$ and $T\geq T_n$, it holds that
		\[
		\E^{\Q'}\Bigl[(1+\Delta A_{T_n}^{\delta})\left(e^{-\int_{(T_n,T]}\Delta V(T_n,u,\delta)\eta(du)}-1\right)\Bigr|\ccF_{T_n-}\Bigr] = 0.
		\]
	\end{enumerate}
\end{corollary}
\begin{proof}
	By means of Bayes' formula, $\Q'$ is an ELMM if and only if $Z'S^\delta P(\cdot,T,\delta)e^{-\int_0^{\cdot}\rois_sds}$ is a local martingale under $\Q$, for every $T\in\R_+$ and $\delta\in\cD_0$.
	The result therefore follows by applying Theorem \ref{thm:HJM} with respect to  $X^0:=e^{\int_0^{\cdot}r^{{\rm OIS}}_sds}/Z'$. By applying Lemma \ref{cor:stoch_exp_property}, we obtain that
	\begin{align*}
		X^0 
		&= \mathcal{E}\bigg(\int_0^{\cdot}\Bigl(\rois_s+\|\theta_s\|^2
		+\int_E\frac{\psi^2(s,x)}{1-\psi(s,x)}\lambda_s(dx)\Bigr)ds 
		+ \theta\cdot W \\
		&\qquad\quad
		+ \frac{\psi}{1-\psi}\ast(\mu-\nu)
		+ \sum_{n\in\N}\frac{Y_n}{1-Y_n}\Ind_{\dbraco{T_n,+\infty}}\bigg).
	\end{align*}
	Note that $\int_0^T\int_E\psi^2(s,x)/(1-\psi(s,x))\lambda_s(dx)ds<+\infty$ a.s., as a consequence of the assumption that $\int_0^T\int_{\{\psi(s,x)\in[0,1]\}}\psi^2(s,x)/(1-\psi(s,x))\lambda_s(dx)ds<+\infty$ a.s.  together with the elementary inequality $x^2/(1-x)\leq|x|\wedge x^2$, for $x\leq 0$.
	The process $X^0$ is of the form \eqref{eq:num}, \eqref{eq:num_FV} with 
	$$
	r_t=\rois_t+\|\theta_t\|^2+\int_E\frac{\psi^2(t,x)}{1-\psi(t,x)}\lambda_t(dx), $$
	$H=\theta$, $L=\psi/(1-\psi)$ and $\Delta B_{T_n}=Y_n/(1-Y_n)$.
	Since 
	$$
	\int_0^T\int_E\frac{\psi^2(s,x)}{1-\psi(s,x)}\lambda_s(dx)ds<+\infty
	$$
	a.s., for all $T>0$, it can be easily checked that condition \eqref{eq:Q_int} is equivalent to \eqref{eq:jumps_int}.
	The corollary then follows from Theorem \ref{thm:HJM} noting that, for any $\ccF_{T_n}$-measurable random variable $\xi$ which is sigma-integrable under $\Q'$ with respect to $\ccF_{T_n-}$, it holds that
	\begin{equation*}
		\E^{\Q'}[\xi|\ccF_{T_n-}] 
		 = \frac{\E^{\Q}[Z'_{T_n}\xi|\ccF_{T_n-}]}{Z'_{T_n-}}
		= \E^{\Q}\bigl[(1-Y_n)\xi\big|\ccF_{T_n-}\bigr]
		 = \E^{\Q}\left[\frac{\xi}{1+\Delta B_{T_n}}\bigg|\ccF_{T_n-}\right],
	\end{equation*}
	where we have used the fact that $Z'_{T_n}=Z'_{T_n-}(1-Y_n)$, for every $n\in\N$.
\end{proof}

\begin{remark}	\label{rem:deflators}
	The proof of Corollary \ref{cor:ELMM_Q} permits to obtain a characterization of all {\em equivalent local martingale deflators} for the multiple curve financial market, i.e., all strictly positive $\Q$-local martingales $Z$ of the form \eqref{eq:stoch_exp_Z} such that $ZS^\delta P(\cdot,T,\delta)e^{-\int_0^{\cdot}\rois_sds}$ is a $\Q$-local martingale, for every $T\in\R_+$ and $\delta\in\cD_0$.
\end{remark}

\begin{remark} \label{rem:CFG}
	The HJM framework of \cite{CuchieroFontanaGnoatto2016}
	can be recovered as a special case with no stochastic discontinuities, setting  $\eta (du) = du$  in \eqref{eq:eta}, taking the OIS bank account as num\'eraire and a jump measure $\mu$ generated by a given It\^o semimartingale. \cite{CuchieroFontanaGnoatto2016} show that  most of the existing multiple curve models are covered by their framework, which a fortiori implies that they can be easily embedded in our framework. 
\end{remark}

 		\section{General market models} \label{sec:market models}

In this section, we consider market models and develop a general arbitrage-free framework for modeling Ibor rates.
As shown in Appendix~\ref{app:MM}, market models can be embedded into the extended HJM framework considered in Section \ref{sec:TSM}, in the spirit of \cite{BGM}. This is possible due to the fact that the measure $\eta(du)$  in the term structure equation \eqref{eq:PtT} may contain atoms. However, it turns out to be simpler to directly study market models as follows.

In the spirit of market models, and differently from Definition \ref{def:market}, in this section we assume that only finitely many assets are traded. For each $\delta\in\cD$, let $\cT^{\delta}=\{T^{\delta}_0,\ldots,T^{\delta}_{N^{\delta}}\}$ be the set of settlement dates of traded FRA contracts associated to tenor $\delta$, with $T^{\delta}_0=T_0$ and $T^{\delta}_{N^{\delta}}=T^*$, for $0\leq T_0<T^*<+\infty$. We consider an equidistant tenor structure, i.e. $T^{\delta}_i-T^{\delta}_{i-1}=\delta$, for all $i=1,\ldots,N^{\delta}$ and $\delta\in\cD$.
Let us also define $\cT:=\medcup_{\delta\in\cD}\cT^{\delta}$, corresponding to the set of all traded FRAs.
The starting point of our approach is representation \eqref{eq:liborFRA},
\beq\label{eq:liborFRA2}
\pifra(t, T, \delta, K) = \delta \big( L(t, T, \delta) - K \big) P(t, T + \delta),
\eeq
for $\delta\in\cD$, $T\in\cT^{\delta}$, $t\in[0,T]$ and $K\in\R$.
The financial market contains OIS zero-coupon bonds for all maturities $T\in\cT^0:=\cT\medcup\{T^*+\delta_i:i=1,\ldots,m\}$\footnote{Note that we need to consider an extended set of maturities for OIS bonds since the payoff of a FRA contract with settlement date $T$ and tenor $\delta$ takes place at date $T+\delta$.} as well as FRA contracts for all $\delta\in\cD$, $T\in\cT^{\delta}$ and $K\in\R$.

Let $(\Omega,\cF,\bbF,\Q)$ be a filtered probability space supporting a $d$-dimensional Brownian motion $W$ and a random measure $\mu$, as described in Section \ref{sec:TSM}.
We assume that, for every tenor $\delta\in\cD$ and maturity $T\in\cT^{\delta}$, the forward Ibor rate $L(\cdot,T,\delta)=(L(t,T,\delta))_{0\leq t\leq T}$ satisfies
\begin{align}
L(t,T,\delta)  &=  L(0,T,\delta) + \int_0^t a^L(s,T,\delta) ds + \sum_{n \in \mathbb{N}} \Delta L(T_n,T,\delta) \ind{ T_n \leq t}  \nonumber\\
&\quad+ \int_0^t b^L(s,T,\delta) dW_s
+ \int_0^t \int_E g^L(s, x, T, \delta) \bigl(\mu(ds,dx) - \nu(ds,dx)\bigr).
\label{def:Ldyn}
\end{align}
In the above equation, $a^L(\cdot,T,\delta)=(a^L(t,T,\delta))_{0\leq t\leq T}$ is a real-valued adapted process  satisfying $\int_0^T|a^L(s,T,\delta)|ds<+\infty$ a.s., $b^L(\cdot,T,\delta)=(b^L(t,T,\delta))_{0\leq t\leq T}$ is a progressively measurable $\R^d$-valued process satisfying the integrability condition $\int_0^T\|b^L(s,T,\delta)\|^2ds<+\infty$ a.s., $(\Delta L(T_n,T,\delta))_{n\in\N}$ is a family of random variables such that $\Delta L(T_n,T,\delta)$ is $\ccF_{T_n}$-measurable, for each $n\in\N$, and $g^L(\cdot,\cdot,T,\delta):\Omega\times[0,T]\times E\rightarrow\R$ is a $\mathcal{P}\otimes\cB_E$-measurable function that satisfies
$$\int_0^T\int_E \left(\bigl(g^L(s,x,T,\delta)\bigr)^2 \wedge |g^L(s,x,T,\delta)|\right) \lambda_s(dx) ds < + \infty \quad \text{a.s.}
$$  
The dates $(T_n)_{n\in\N}$ represent the stochastic discontinuities occurring in the market.
We assume that OIS bond prices are of the form \eqref{eq:PtT} for $\delta=0$, for all $T\in\cT^0$, with the associated forward rates $f(t,T,0)$ being as in \eqref{eq:fwd_rate}.

The main goal of this section consists in deriving necessary and sufficient conditions for a reference probability measure $\Q$ to be a risk-neutral measure with respect to a general num\'eraire $X^0$ of the form \eqref{eq:num} for the financial market where FRA contracts and OIS zero-coupon bonds are traded, and FRA prices are modeled via \eqref{eq:liborFRA2} and \eqref{def:Ldyn} for the discrete set $\cT$ of settlement dates. 
We recall that 
\begin{align*}
\bar{b}(t, T+\delta, 0) & = \int_{[t, T+\delta]} b(t, u, 0) \eta(du), \\
\bar{g}(t, x, T+\delta, 0) & = \int_{[t, T+\delta]} g(t, x, u, 0) \eta(du).
\end{align*}

\begin{theorem}\label{thm:Libor}
	Suppose that Assumption \ref{ass} holds for $\delta=0$ and for all maturities $T\in\cT^0$.
	Then $\Q$ is a risk-neutral measure with respect to $X^0$ if and only if all the conditions of Theorem \ref{thm:HJM} are satisfied for $\delta=0$ and for all  $T\in\cT^0$, and, for every $\delta\in\cD$,
	\beq	\label{eq:jumps_int_lib}
	\int_0^T \int_E \Bigl|g^L(s,x,T,\delta)\left(\frac{e^{-\bar{g}(s,x,T + \delta , 0)}}{1+L(s,x)} - 1\right)\Bigr|\lambda_s(dx)ds  < +\infty
	\eeq
	a.s. for all $T\in\cT^{\delta}$, and, for each $n\in\N$ and $\cT^{\delta}\owns T\geq T_n$, the random variable
	\beq	\label{eq:jump_sigma_int_lib}
	\frac{\Delta L(T_n,T,\delta)}{1+\Delta B_{T_n}}e^{-\int_{(T_n,T + \delta]}\Delta V(T_n,u,0)\eta(du)}
	\eeq
	is sigma-integrable with respect to $\ccF_{T_n-}$, and the following two conditions hold a.s.:
	\begin{enumerate}
		\item[(i)]
		for all $T\in\cT^{\delta}$ and a.e. $t\in[0,T]$, it holds that
		\begin{align*}
		a^L(t,T,\delta) & = b^L(t, T, \delta)^\top\bigl(H_t + \bar{b}(t, T + \delta, 0)\bigr)\\		
		&\quad - \int_E g^L(t, x, T, \delta) 
		\left(\frac{e^{-\bar{g}(t, x, T + \delta, 0)}}{1+L(t, x)}   - 1 \right)  \lambda_t(dx);			\end{align*}
		\item[(ii)]
		for all $n\in\N$ and $\cT^{\delta}\owns T\geq T_n$, it holds that
		\[
		\E^{\Q}\left[\frac{\Delta L(T_n,T,\delta) }{1 + \Delta B_{T_n}}e^{-\int_{(T_n,T + \delta]}\Delta V(T_n, u, 0)\eta(du)} \biggr| \ccF_{T_n -}\right] = 0.
		\]
	\end{enumerate}
\end{theorem}

Condition {\em (i)} of Theorem \ref{thm:Libor} is a drift restriction for the Ibor rate process. In the context of a continuum of traded maturities, as in Theorem \ref{thm:HJM}, this condition can be separated into a condition on the short end and an  HJM-type drift restriction (see conditions {\em (i)} and {\em (ii)} in Theorem \ref{thm:HJM}). Condition {\em (ii)}, similarly to conditions {\em (iii)}, {\em (iv)} of Theorem \ref{thm:HJM}, corresponds to requiring that, for each $n\in\N$, the size of the jumps occurring at date $T_n$ in FRA prices cannot be predicted on the basis of the information contained in $\cF_{T_n-}$.

\begin{proof}
	In view of representation \eqref{eq:liborFRA2}, $\Q$ is a risk-neutral measure with respect to $X^0$ if and only if $P(\cdot,T)/X^0$ is a $\Q$-local martingale, for every $T\in\cT^0$, and $L(\cdot,T,\delta)P(\cdot,T+\delta)/X^0$ is a $\Q$-local martingale, for every $\delta\in\cD$ and $T\in\cT^{\delta}$.
	Considering first the OIS bonds, Theorem \ref{thm:HJM} implies that $P(\cdot,T)/X^0$ is a $\Q$-local martingale, for every $T\in\cT^{0}$, if and only if conditions \eqref{eq:jumps_int}, \eqref{eq:jump_sigma_int} as well as conditions {\em (i)}--{\em (iv)} of Theorem \ref{thm:HJM} are satisfied for $\delta=0$ and for all  $T\in\cT^0$.
	Under these conditions, equation \eqref{eq:SPX_pred} for $\delta=0$ gives that 
	\beq	\label{eq:PX}
	\frac{P(\cdot,T)}{X^0} = P(0,T)\,
	\cE\bigl(M'(T,0)\bigr),
	\eeq
	for every $T\in\cT^{0}$, where the local martingale $M'(T,0)$ is given by
	\begin{align*}
	M'(T,0)
	&= K^{(2)}(T,0)
	- \int_0^{\cdot}\bigl(H_s+\bar{b}(s,T,0)\bigr)dW_s	\\
	&\quad + \int_0^{\cdot}\int_E \left(\frac{e^{-\bar{g}(s,x,T,0)}}{1+L(s,x)}-1\right)\bigl(\mu(ds,dx)-\nu(ds,dx)\bigr),
	\end{align*}
	as follows from equation \eqref{eq:comp_0}, with
	\[
	K^{(2)}(T,0) = \sum_{n\in\N}\left(\frac{e^{-\int_{(T_n,T]}\Delta V(T_n,u,0)\eta(du)+f(T_n-,T_n,0)}}{1+\Delta B_{T_n}}-1\right)\Ind_{\dbraco{T_n,+\infty}}.
	\]
	By relying on \eqref{def:Ldyn} and \eqref{eq:PX}, we can compute
	\begin{align}
	&d\left(L(t,T,\delta)\frac{P(t,T+\delta)}{X^0_t}\right)	\notag
	= \frac{P(t-,T+\delta)}{X^0_{t-}}\\
	&\cdot\biggl(dL(t,T,\delta) + L(t-,T,\delta)d M'_t(T+\delta,0) + d\bigl[L(\cdot,T,\delta),M'(T+\delta,0)\bigr]_t\biggr)	\notag\\
	& = \frac{P(t-,T+\delta)}{X^0_{t-}}\biggl(M''_t(T,\delta) + j_t(T,\delta)dt + dJ^{(1)}_t(T,\delta) + dJ^{(2)}_t(T,\delta)\biggr),
	\label{eq:LPX}
	\end{align}
	where $M''(T,\delta)=(M''_t(T,\delta))_{0\leq t\leq T}$ is a local martingale given by
	\begin{align*}
	M''_t(T,\delta) &:= \int_0^tL(s-,T,\delta)dM'_s(T+\delta,0) + \int_0^tb^L(s,T,\delta)dW_s \\
	&\quad + \int_0^t\int_E g^L(s,x,T,\delta)\bigl(\mu(ds,dx)-\nu(ds,dx)\bigr),
	\end{align*}
	$j(T,\delta)=(j_t(T,\delta))_{0\leq t\leq T}$ is an adapted real-valued process given by
	\[
	j_t(T,\delta) = a^L(t,T,\delta) - b^L(t,T,\delta)^{\top}\bigl(H_t+\bar{b}(t,T+\delta,0)\bigr),
	\]
	$J^{(1)}(T,\delta)=(J^{(1)}_t(T,\delta))_{0\leq t\leq T}$ is a pure jump finite variation process given by
	\[
	J^{(1)}_t(T,\delta) = \int_0^t\int_E g^L(s,x,T,\delta)\left(\frac{e^{-\bar{g}(s,x,T+\delta,0)}}{1+L(s,x)}-1\right)\mu(ds,dx),
	\]
	and $J^{(2)}(T,\delta)=(J^{(2)}_t(T,\delta))_{0\leq t\leq T}$ is a pure jump finite variation process given by
	\[
	J^{(2)}_t(T,\delta) = \sum_{n\in\N}\ind{T_n\leq t}\frac{\Delta L(T_n,T,\delta)}{1+\Delta B_{T_n}}e^{-\int_{(T_n,T+\delta]}\Delta V(T_n,u,0)\eta(du)+f(T_n-,T_n,0)}.
	\]
	If $L(\cdot,T,\delta)P(\cdot,T+\delta)/X^0$ is a local martingale, for every $\delta\in\cD$ and $T\in\cT^{\delta}$, then \eqref{eq:LPX} implies that the processes $J^{(1)}(T,\delta)$ and $J^{(2)}(T,\delta)$ are of locally integrable variation. Similarly as in the proof of Theorem \ref{thm:HJM}, this implies the validity of conditions \eqref{eq:jumps_int_lib} and \eqref{eq:jump_sigma_int_lib}, due to Theorem 5.29 in \cite{he1992semimartingale}. 
	Let us denote by $\widehat{J}^{(i)}(T,\delta)$ the compensator of $J^{(i)}(T,\delta)$, for $i\in\{1,2\}$, $\delta\in\cD$ and $T\in\cT^{\delta}$.
	We have that
	\begin{align*}
	\widehat{J}^{(1)}(T,\delta) &= \int_0^{\cdot}\int_E g^L(s,x,T,\delta)\left(\frac{e^{-\bar{g}(s,x,T+\delta,0)}}{1+L(s,x)}-1\right)\lambda_s(dx)ds,	\\
	\widehat{J}^{(2)}(T,\delta) &= \sum_{n\in\N}\bigg(\E^{\Q}\left[\frac{\Delta L(T_n,T,\delta)}{1+\Delta B_{T_n}}e^{-\int_{(T_n,T+\delta]}\Delta V(T_n,u,0)\eta(du)}\bigg|\ccF_{T_n-}\right] e^{f(T_n-,T_n,0)}\Ind_{\dbraco{T_n,+\infty}}\bigg).
	\end{align*}
	The local martingale property of $L(\cdot,T,\delta)P(\cdot,T+\delta)/X^0$ together with equation \eqref{eq:LPX} implies that the predictable finite variation process
	\beq	\label{eq:comp_LPX}
	\int_0^{\cdot}j_s(T,\delta)ds + \widehat{J}^{(1)}(T,\delta) + \widehat{J}^{(2)}(T,\delta)
	\eeq
	is null (up to an evanescent set), for every $\delta\in\cD$ and $T\in\cT^{\delta}$. Considering separately the absolutely continuous and discontinuous parts, this implies the validity of conditions {\em (i)}, {\em (ii)} in the statement of the theorem.
	
	Conversely, by Theorem \ref{thm:HJM}, if conditions \eqref{eq:jumps_int}, \eqref{eq:jump_sigma_int} as well as conditions {\em (i)}--{\em (iv)} of Theorem \ref{thm:HJM} are satisfied for $\delta=0$ and for all  $T\in\cT^0$, then $P(\cdot,T)/X^0$ is a $\Q$-local martingale, for all $T\in\cT^0$. Furthermore, if conditions \eqref{eq:jumps_int_lib}, \eqref{eq:jump_sigma_int_lib} are satisfied and conditions {\em (i)}, {\em (ii)} of the theorem hold, then the process given in \eqref{eq:comp_LPX} is null. In turn, by equation \eqref{eq:LPX}, this implies that $L(\cdot,T,\delta)P(\cdot,T+\delta)/X^0$ is a $\Q$-local martingale, for every $\delta\in\cD$ and $T\in\cT^{\delta}$, thus proving that $\Q$ is a risk-neutral measure with respect to $X^0$.
\end{proof}

\begin{remark}\label{rem:termBond}
	In  market models, the num\'eraire is usually chosen as the OIS zero-coupon bond with the longest available maturity  $T^*$ ({\em terminal bond}). In addition, the reference probability measure $\Q$ is the associated $T^*$-forward  measure, see Section 12.4 in \cite{MusielaRutkowski}. 
	Exploiting the generality of the process $X^0$, this setting can be easily accommodated within our framework. Indeed, if $\int_0^{T^*}\int_E|e^{-\bar{g}(s,x,T^*,0)}-1+\bar{g}(s,x,T^*,0)|\lambda_s(dx)ds<+\infty$ a.s., Corollary \ref{cor:bond_stoch_exp} shows that $X^0=P(\cdot,T^*)/P(0,T^*)$ holds as long as the processes appearing in \eqref{eq:num} and \eqref{eq:num_FV} are specified as
	\begin{align*}
	H_t &= -\bar{b}(t,T^*,0),\\
	L(t,x) &= e^{-\bar{g}(t,x,T^*,0)}-1,\\
	\Delta B_{T_n} &= e^{-\int_{(T_n,T^*]}\Delta V(T_n,u,0)\eta(du)+f(T_n-,T_n,0)}-1,\\
	r_t &= f(t,t,0)-\bar{a}(t,T^*,0)+\frac{1}{2}\|\bar{b}(t,T^*,0)\|^2\\
	&\quad+\int_E\bigl(e^{-\bar{g}(t,x,T^*,0)}-1+\bar{g}(t,x,T^*,0)\bigr)\lambda_t(dx).
	\end{align*}
	Under this specification, a direct application of Theorem \ref{thm:Libor} yields necessary and sufficient conditions for $\Q$ to be a risk-neutral measure with respect to the terminal OIS bond as num\'eraire.
\end{remark}

\subsection{Martingale modeling}
Typically, market models start directly from the assumption that each Ibor rate $L(\cdot,T,\delta)$ is a martingale under the $(T+\delta)$-forward measure $\Q^{T+\delta}$ associated to the num\'eraire $P(\cdot,T+\delta)$.
In our context, this assumption is generalized into a {\em local martingale} requirement under the $(T+\delta)$-forward measure, whenever the latter is well-defined. 
More specifically, suppose that $P(\cdot,T+\delta)/X^0$ is a true martingale and define the $(T+\delta)$-forward measure by $d\Q^{T+\delta}|_{\ccF_{T+{\delta}}}:=(P(0,T+\delta)X^0_{T+\delta})^{-1}d\Q|_{\ccF_{T+{\delta}}}$.
As a consequence of Girsanov's theorem (see \cite[Theorem III.3.24]{JacodShiryaev}) and equation \eqref{eq:PX}, the forward Ibor rate $L(\cdot,T,\delta)$ satisfies under the measure $\Q^{T+\delta}$
\begin{align}
L(t,T,\delta) & = L(0,T,\delta) + \int_0^t a^{L,T+\delta}(s,T,\delta) ds + \sum_{n \in \mathbb{N}} \Delta L(T_n,T,\delta) \ind{ T_n \leq t}  \nonumber\\
&\quad 
+ \int_0^t b^L(s,T,\delta) dW^{T+\delta}_s+ \int_0^t \int_E g^L(s, x, T, \delta) \bigl(\mu(ds,dx) - \nu^{T+\delta}(ds,dx)\bigr),
\label{def:Ldyn_fwdmeas}
\end{align}
for some adapted real-valued process $a^{L,T+\delta}(\cdot,T,\delta)$, where the process $W^{T+\delta}$ is a $\Q^{T+\delta}$-Brownian motion defined by $W^{T+\delta}:=W+\int_0^{\cdot}(H_s+\bar{b}(s,T+\delta,0))ds$ and the compensator $\nu^{T+\delta}(ds,dx)$ of the random measure $\mu(ds,dx)$ under $\Q^{T+\delta}$ is given by
\[
\nu^{T+\delta}(ds,dx) = \frac{e^{-\bar{g}(s,x,T+\delta,0)}}{1+L(s,x)}\lambda_s(dx)ds.
\]
In this context, Theorem \ref{thm:Libor} leads to the following proposition, which provides a characterization of the local martingale property of forward Ibor rates under forward measures.

\begin{proposition}\label{prop:Liborchangeofmeasure}
	Suppose that Assumption \ref{ass} holds for $\delta=0$ and for all $T\in\cT^0$. Assume furthermore that $P(\cdot,T)/X^0$ is a true $\Q$-martingale, for every $T\in\cT^0$. 
	Then the following are equivalent:
	\begin{enumerate}
		\item[(i)]
		$\Q$ is a risk-neutral measure;
		\item[(ii)]
		$L(\cdot,T,\delta)$ is a local martingale under $\Q^{T+\delta}$, for every $\delta\in\cD$ and $T\in\cT^{\delta}$;
		\item[(iii)]
		for every $\delta\in\cD$ and $T\in\cT^{\delta}$, it holds that
		\[
		a^{L,T+\delta}(t,T,\delta) = 0,
		\]
		outside a subset of $\Omega\times[0,T]$ of $(\Q\otimes dt)$-measure zero, and, for every $n\in\N$ and $\cT^{\delta}\owns T\geq T_n$, the random variable $\Delta L(T_n,T,\delta)$ 
		satisfies
		\[
		\E^{\Q^{T+\delta}}\left[\Delta L(T_n,T,\delta)|\ccF_{T_n-}\right] = 0
		\text{ a.s.}
		\]
	\end{enumerate}
\end{proposition}
\begin{proof}
	Under these assumptions, $\Q$ is a risk-neutral measure if and only if $L(\cdot,T,\delta)P(0,T+\delta)/X^0$ is a local martingale under $\Q$, for every $\delta\in\cD$ and $T\in\cT^{\delta}$. The equivalence $(i)\Leftrightarrow(ii)$ then follows from the conditional version of Bayes' rule (see \cite[Proposition III.3.8]{JacodShiryaev}), while the equivalence $(ii)\Leftrightarrow(iii)$ is a direct consequence of equation \eqref{def:Ldyn_fwdmeas} together with \cite[Theorem 5.29]{he1992semimartingale}.
\end{proof}
	
\section{Affine specifications}\label{sec:affine}

One of the most successful classes of processes in term-structure modeling is the class of affine processes. This class combines a great flexibility in capturing the important features of interest rate markets with a remarkable analytical tractability, see e.g. \cite{DuffieKan}, \cite{DuffieFilipovicSchachermayer}, as well as \cite{Filipovic2009} for a textbook account.
In the literature, affine processes are by definition stochastically continuous and, therefore, do not allow for jumps at predetermined dates. In view of our modeling objectives, we need a suitable generalization of the notion of affine process. To this effect, \cite{KellerResselSchmidtWardenga2018} have recently introduced {\em affine semimartingales} by dropping the requirement of stochastic continuity. Related results on affine processes with stochastic discontinuities in credit risk may be found in \cite{GehmlichSchmidt2016MF}. In the present section, we aim at showing how the class of affine semimartingales leads to flexible and tractable multiple curve models with stochastic discontinuities.

We consider a countable set $\bbT=\{T_n : n\in\N \}$ of discontinuity dates, with $T_{n+1}>T_n$, for every $n\in\N$, and $\lim_{n\rightarrow+\infty}T_n=+\infty$.
We assume that the filtered probability space $(\Omega,\cF,\bbF,\Q)$ supports a $d$-dimensional special semimartingale $X=(X_t)_{t\geq0}$ which is further assumed to be an {\em affine semimartingale} in the sense of \cite{KellerResselSchmidtWardenga2018} and to admit the canonical decomposition
\begin{equation*} 
X = X_0 + B^X+ X^c + x\ast\left(\mu^X-\nu^X\right),
\end{equation*}
where $B^X$ is a finite variation predictable process, $X^c$ is a continuous local martingale with quadratic variation $C^X$ and $\mu^X-\nu^X$ is the compensated jump measure of $X$. 
Let $B^{X,c}$ be the continuous part of $B^X$ and $\nu^{X,c}$ the continuous part of the random measure $\nu^X$, in the sense of \cite[\textsection~II.1.23]{JacodShiryaev}.
In view of \cite[Theorem 3.2]{KellerResselSchmidtWardenga2018}, under weak additional assumptions it holds that
\begin{equation}\label{affsemchar}\begin{aligned} 
B^{X,c}_t(\omega) &= \int_0^t \Big( \beta_0 (s)  + \sum_{i=1}^d X^i_{s-}(\omega) \beta_i(s)   \Big) ds, \\
C^{X}_t(\omega) &= \int_0^t \Big( \alpha_0 (s) + \sum_{i=1}^d X^i_{s-}(\omega) \alpha_i (s)  \Big) ds, \\
\nu^{X,c}(\omega,dt,dx) &= \Bigl(\mu_0(t,dx)+\sum_{i=1}^d X^i_{t-}(\omega)\mu_i(t,dx)\Bigr)dt,\\
\int_{\R^d} \bigl(e^{\langle u,x \rangle}  - 1\bigr) \nu^X(\omega,\{t\},dx) &= \left( \exp\Big( \gamma_0(t,u) + \sum_{i=1}^d \langle X^i_{t-}(\omega), \gamma_i (t,u) \rangle \Big) - 1\right).
\end{aligned} \end{equation}
In \eqref{affsemchar}, we have that $\beta_i:\R_+\rightarrow\R^d$ and $\alpha_i:\R_+\rightarrow\R^{d\times d}$, for $i=0,1,\ldots,d$, $\gamma_0:\R_+\times\mathbb{C}^d\rightarrow\mathbb{C}_{-}$ and $\gamma_i:\R_+\times\mathbb{C}^d\rightarrow\mathbb{C}^d$, for $i=1,\ldots,d$.
$\mu_i(t,dx)$ is a Borel measure on $\R^d\setminus\{0\}$ for all  $i=0,1,\ldots,d$, such that for all $t\in\R_+$, $\int_{\R^d\setminus\{0\}}(1+|x|^2)\mu_i(t,dx)<+\infty$.
Finally, we assume that $\nu^X(\{t\}\times\R^d)$ vanishes a.s. outside the set of stochastic discontinuities $(T_n)_{n\in\N}$.

We use the affine semimartingale $X$ as the driving process of a multiple curve model, as presented in Section \ref{sec:TSM}. In particular, we focus here on modeling the $\delta$-tenor bond prices $P(t,T,\delta)$ and the multiplicative spreads $S^{\delta}_t$ in such a way that the resulting model is {\em affine} in the sense of the following definition, which extends the approach of \cite[Section 5.3]{KellerResselSchmidtWardenga2018}.

\begin{definition}  \label{def:affine}
	The multiple curve model is said to be {\em affine} if
	\begin{align}\label{eq:fS}
	f(t,T,\delta) &= f(0,T,\delta) + \int_0^t\varphi(s,T,\delta) dX_s,
	& \text{ for all }\delta\in\cD_0, \\
	S_t^\delta &= S^{\delta}_0\exp\left(\int_0^t\psi^\delta_s dX_s\right),
	& \text{ for all }\delta\in\cD,
	\end{align}
	for all $0\leq t\leq T<+\infty$, where $\varphi:\Omega\times\R^2_+\times\cD_0\rightarrow\R^d$ and $\psi^\delta:\Omega\times\R_+\times\cD\rightarrow\R^d$ are predictable processes such that, for every $i=1,\ldots,d$ and $T\in\R_+$,
	\[
	\psi^\delta\in L(X)
	\qquad\text{and}\qquad
	\int_0^T|\psi^\delta_t||dB^{X,c}_t|<+\infty\text{ a.s.},
	\quad\text{ for all }\delta\in\cD,
	\]
	and, for all $\delta\in\cD_0$ and $T\in\R_+$,
	\begin{align*}
	\left(\int_0^T|\varphi^i(\cdot,u,\delta)|^2\eta(du)\right)^{1/2}&\in L(X^i)
	\qquad\text{and}\qquad \\
	\int_0^T\int_0^T|\varphi(t,u,\delta)|\eta(du)|dB^{X,c}_t|&<+\infty\text{ a.s.},		
	\end{align*}
	with $L(X)$ denoting the set of $\R^d$-valued predictable processes which are integrable with respect to $X$ in the semimartingale sense, and similarly for $L(X^i)$.
	The measure $\eta$ is specified as in equation \eqref{eq:eta}.
\end{definition}

For all $0 \leq t \leq T <+\infty$ and $\delta\in\cD_0$, let us also define
\begin{align*}
\bar \varphi(t,T,\delta) := \int_{[t,T]}\varphi(t,u,\delta) \eta(du).
\end{align*}
We furthermore assume that $\int_0^Te^{(\psi^{\delta}_t)^\top x}\ind{(\psi^{\delta}_t)^{\top}x>1}\nu^{X,c}(dt,dx)<+\infty$ a.s., for all $T\in\R_+$, which ensures that $S^{\delta}$ is a special semimartingale (see \cite[Proposition II.8.26]{JacodShiryaev}).
To complete the specification of the model, we suppose that $X^0$ takes the form
\begin{equation}  \label{eq:num_affine}
X^0_t = \exp\Big(\int_0^t r_s ds + \sum_{n\in\N} \psi_{T_n}^\top \Delta X_{T_n}\ind{T_n\leq t} \Big),
\qquad\text{ for all }t\geq0,
\end{equation}
where $(r_t)_{t\geq0}$ is an adapted real-valued process satisfying $\int_0^T|r_t|dt<+\infty$ a.s., for all $T\in\R_+$, and $\psi_{T_n}$ is a $d$-dimensional $\cF_{T_n-}$-measurable random vector, for all $n\in\N$.

We aim at characterizing when $\Q$ is a risk-neutral measure for an affine multiple curve model. By Remark \ref{rem:condition1}, we see that a necessary condition is that
\begin{equation}  \label{eq:rf}
r_t = f(t,t,0),
\qquad\text{ for a.e. }t\geq0.
\end{equation}

Under the present assumptions and in the spirit of Theorem \ref{thm:HJM}, the following proposition provides sufficient conditions for $\Q$ to be a risk-neutral measure for the affine multiple curve model introduced above. For convenience of notation we let $\psi^0_t := 0$ for all $t\in \R_+$ and $S^0_0 := 1$, so that $S^0:= S^0_0\exp(\int_0^\cdot\psi^0_s dX_s) \equiv 1$.

\begin{proposition}\label{prop:affine}
	Consider an affine multiple curve model as in Definition \ref{def:affine} and satisfying \eqref{eq:rf}. Assume furthermore that
	\begin{equation}  \label{eq:aff_integrability}
	\int_0^T\int_{\R^d\setminus\{0\}}\left|e^{(\psi^\delta_s)^{\top}x}\bigl(e^{-\bar{\varphi}(s,T,\delta)^{\top}x}-1\bigr)+\bar{\varphi}(s,T,\delta)^{\top}x\right|\nu^{X,c}(ds,dx)<+\infty
	\quad\text{a.s.}
	\end{equation}
	for every $\delta\in\cD_0$ and $T\in\R_+$.
	Then $\Q$ is a risk-neutral measure with respect to $X^0$ given as in \eqref{eq:num_affine} if the following three conditions hold a.s. for every $\delta\in\cD_0$:
	\begin{enumerate}[(i)]
		\item 
		for a.e. $t\in\R_+$, it holds that
		\begin{align*}
		r_t - f(t,t,\delta) &= 
		(\psi^\delta_t)^\top\left( \beta_0(t)+\sum_{i=1}^d X^i_{t-}\beta_i(t) \right) + \frac{1}{2}(\psi^\delta_t)^{\top}\left( \alpha_0(t)+\sum_{i=1}^d X^i_{t-}\alpha_i(t) \right)\psi^\delta_t \\
		& \quad+ \int_{\R^d\setminus\{0\}}\left(e^{(\psi^\delta_t)^{\top}x}-1-(\psi^\delta_t)^{\top}x\right)\left(\mu_0(t,dx)+\sum_{i=1}^dX^i_{t-}\mu_i(t,dx)\right);
		\end{align*}
		\item 
		for every $T\in\R_+$,  a.e. $t\in[0,T]$ and for every $i=0,1,\ldots,d$, it holds that
		\begin{align} \label{eq:affineI}
		\bar \varphi(t,T,\delta)^\top \beta_i(t) 
		&= \bar \varphi(t,T,\delta)^\top \alpha_i(t) \left( \half \bar \varphi(t,T,\delta ) - \psi^\delta_t\right)  \notag\\
		&\quad+ \int_{\R^d\setminus\{0\}}\left(e^{(\psi^\delta_t)^{\top}x}\left(e^{-\bar{\varphi}(t,T,\delta)^{\top}x}-1\right)+\bar{\varphi}(t,T,\delta)^{\top}x\right)\mu_i(t,dx);
		\end{align}
		\item
		for every $n\in\N$ and $T\geq T_n$, it holds that
		\begin{align*}
		-f(T_n-,T_n,\delta)
		&=\gamma_0\Bigl(T_n,\psi^{\delta}_{T_n}-\psi_{T_n}-\int_{(T_n,T]}\varphi(T_n,u,\delta)\eta(du)\Bigr) \\
		&\quad+ \sum_{i=1}^d\Bigl\langle X^i_{T_n-},\gamma_i\Bigl(T_n,\psi^\delta_{T_n}-\psi_{T_n}-\int_{(T_n,T]}\varphi(T_n,u,\delta)\eta(du)\Bigr)\Bigr\rangle.
		\end{align*}
	\end{enumerate}
\end{proposition}

\begin{proof}
	For all $\delta\in\cD_0$, the present integrability assumptions ensure that $\psi^\delta \cdot X$ and $S^{\delta}$ are special semimartingales.
	Hence, \cite[Theorem II.8.10]{JacodShiryaev} implies that  $S^{\delta}$ admits a stochastic exponential representation of the form \eqref{eq:spread}, \eqref{eq:spread_FV}, with
	\begin{align*}
	\alpha^{\delta}_t
	&= (\psi^\delta_t)^\top\left( \beta_0(t)+\sum_{i=1}^d X^i_{t-}\beta_i(t) \right)
	+ \frac{1}{2}(\psi^\delta_t)^{\top}\left( \alpha_0(t)+\sum_{i=1}^d X^i_{t-}\alpha_i(t) \right)\psi^\delta_t  \\
	&\quad+ \int_{\R^d\setminus\{0\}}\left(e^{(\psi^\delta_t)^{\top}x}-1-(\psi^\delta_t)^{\top}x\right) \left(\mu_0(t,dx)+\sum_{i=1}^dX^i_{t-}\mu_i(t,dx)\right),\\
	\Delta A^{\delta}_{T_n} &= e^{(\psi^\delta_{T_n})^{\top}\Delta X_{T_n}}-1,
	\qquad\text{ for all }n\in\N,
	\end{align*}
	and $L^{\delta}(t,x) = (e^{(\psi^\delta_t)^{\top}x}-1)\Ind_{J^c}(t)$, for all $(t,x)\in\R_+\times\R^d\setminus\{0\}$, where we define the set $J^c:=\R_+\setminus\bbT$.
	Due to \eqref{eq:num_affine}, condition {\em (i)} of Theorem \ref{thm:HJM} reduces to $a^{\delta}_t=f(t,t,0)-f(t,t,\delta)$, for a.e. $t\in\R_+$ and $\delta\in\cD$ (see also equation \eqref{eq:HJM_s} in Remark \ref{rem:condition1}), from which condition {\em (i)} directly follows.
	The integrability conditions appearing in Definition \ref{def:affine} enable us to apply the stochastic Fubini theorem in the version of Theorem IV.65 of \cite{Protter} and, moreover, ensure that $\varphi(\cdot,T,\delta)\cdot X$ is a special semimartingale, for every $\delta\in\cD_0$ and $T\in\R_+$. This permits to obtain a representation of $P(t,T,\delta)$ as in Lemma \ref{lem:bond_prices}, namely
	\begin{align*}
	P(t,T,\delta) &= \exp\bigg(-\int_0^Tf(0,u,\delta)\eta(du) - \int_0^t\bar{\varphi}(s,T,\delta)dB^{X,c}_s \\
	&\qquad\qquad - \sum_{n\in\N}\bar{\varphi}(T_n,T,\delta)^{\top}\Delta X_{T_n}\ind{T_n\leq t} 
	- \int_0^t\bar{\varphi}(s,T,\delta)dX^c_s 
	\\
	&\qquad\qquad 	-\int_0^t\int_{\R^d\setminus\{0\}}\bar{\varphi}(s,T,\delta)^{\top}x\Ind_{J^c}(s)\bigl(\mu^X(ds,dx)-\nu^X(ds,dx)\bigr)\\
	&\qquad\qquad + \int_0^tf(u,u,\delta)\eta(du)\bigg).
	\end{align*}
	In view of the affine structure \eqref{affsemchar} and comparing with \eqref{eq:dec_PtT}, it holds that
	\begin{align*}
	\bar{a}(t,T,\delta) &= \bar{\varphi}(t,T,\delta)^{\top}\Bigl(\beta_0 (t)  + \sum_{i=1}^d X^i_{t-} \beta_i(t)\Bigr),\\
	\|\bar{b}(t,T,\delta)\|^2 &= \bar{\varphi}(t,T,\delta)^{\top} \Bigl(\alpha_0 (t)  + \sum_{i=1}^d X^i_{t-} \alpha_i(t)\Bigr)\bar{\varphi}(t,T,\delta),\\
	\bar{b}(t,T,\delta)^{\top}H^{\delta}_t &= \bar{\varphi}(t,T,\delta)^{\top}\Bigl(\alpha_0 (t)  + \sum_{i=1}^d X^i_{t-} \alpha_i(t)\Bigr)\psi^\delta_t,
	\end{align*}
	and $\bar{g}(t,x,T,\delta) = \bar{\varphi}(t,T,\delta)^{\top}x\Ind_{J^c}(t)$, for all $0\leq t\leq T<+\infty$, $\delta\in\cD_0$ and $x\in\R^d\setminus\{0\}$.
	In the present setting condition {\em (ii)} of Theorem \ref{thm:HJM} takes the form
	\begin{align}\label{eq:aff_proof} 
	&\bar{\varphi}(t,T,\delta)^{\top}\biggl(\beta_0 (t)  + \sum_{i=1}^d X^i_{t-} \beta_i(t)\biggr)= \bar{\varphi}(t,T,\delta)^{\top}\biggl(\alpha_0 (t)  + \sum_{i=1}^d X^i_{t-} \alpha_i(t)\biggr)\left(\frac{1}{2}\bar{\varphi}(t,T,\delta)-\psi^\delta_t\right) \nonumber\\
	&\quad
	+ \int_{\R^d\setminus\{0\}}	\left(e^{(\psi^\delta_t)^{\top}x}\bigl(e^{-\bar{\varphi}(t,T,\delta)^{\top}x}-1\bigr)+\bar{\varphi}(t,T,\delta)^{\top}x\right)\Bigl(\mu_0(t,dx)+\sum_{i=1}^d X^i_{t-}\mu_i(t,dx)\Bigr).
	\end{align}
	Clearly, condition {\em (ii)} of the proposition is sufficient for \eqref{eq:aff_proof} to hold, for every $T\in\R_+$ and a.e. $t\in[0,T]$.
	In the present setting, conditions {\em (iii)}, {\em (iv)} of Theorem \ref{thm:HJM} can be together rewritten as follows, for every $\delta\in\cD_0$, $n\in\N$ and $T\geq T_n$,
	\begin{align*}
	e^{-f(T_n-,T_n,\delta)}
	&= \E^{\Q}\left[\frac{1+\Delta A^{\delta}_{T_n}}{1+\Delta B_{T_n}}e^{-\int_{(T_n,T]}\varphi(T_n,u,\delta)^{\top}\Delta X_{T_n}\,\eta(du)}\biggl|\cF_{T_n-}\right] \\
	&= \E^{\Q}\left[\exp\left(\biggl(\psi^\delta_{T_n}-\psi_{T_n}-\int_{(T_n,T]}\varphi(T_n,u,\delta)\eta(du)\biggr)^{\top}\Delta X_{T_n}\right)\biggl|\cF_{T_n-}\right],
	\end{align*}
	from which condition {\em (iii)} of the proposition follows by making use of \eqref{affsemchar}.
	Finally, in the present setting the integrability condition \eqref{eq:jumps_int} appearing in Theorem \ref{thm:HJM} reduces to condition \eqref{eq:aff_integrability}. In view of Theorem \ref{thm:HJM}, we can conclude that $\Q$ is a risk-neutral with respect to $X^0$.
\end{proof}

\begin{remark}
	Condition {\em (ii)} is only sufficient for the necessary condition \eqref{eq:aff_proof}. Only if the coordinates of $X^i$ are linearly independent, then this condition is also necessary.
\end{remark}

The following examples illustrate the conditions of Proposition \ref{prop:affine}.

\begin{example}[A single-curve Vasi\v cek specification]\label{ex:Vasicek}
	As  first example we study a classical single-curve (i.e., $\cD=\emptyset$) model without jumps, driven by a one-dimensional Gaussian Ornstein-Uhlenbeck process. 
	Let  $\xi$ be the solution of 
	$$ 
	d\xi_t = \kappa (\theta - \xi_t )dt + \sigma dW_t, 
	$$
	where $W$ is a Brownian motion and $\kappa, \theta, \sigma$ are positive constants. 
	As driving process in \eqref{eq:fS} we choose the three-dimensional affine process 
	$$
	X_t=\left(t,\int_0^t \xi_sds, \xi_t\right)^\top, \ t \ge 0.
	$$ 
	The coefficients in the affine semimartingale representation \eqref{affsemchar} are time-homogeneous, i.e.~$\alpha_i(t)=\alpha_i$ and $\beta_i(t)=\beta_i$, $i=0,\dots,3,$ given by
	$$ \beta_0 = 
	\left( \begin{array}{c} 
	1 \\ 0 \\ \kappa \theta 
	\end{array}\right), 
	\ \beta_1 = \left( \begin{array}{c}
	0 \\0 \\ 0
	\end{array}\right), 
	\ \beta_2 = \left( \begin{array}{c}
	0 \\0 \\ 0
	\end{array}\right), 
	\ \beta_3 = \left( \begin{array}{c}
	0 \\1 \\ -\kappa
	\end{array}\right), 
	\ \alpha_0 = \left( \begin{array}{ccc}
	0 & 0 & 0 \\ 0 & 0 & 0 \\ 0 & 0 & \sigma^2
	\end{array}\right), 
	$$
	and $\alpha_1=\alpha_2=\alpha_3=0$. 
	The drift condition \eqref{eq:affineI}  implies  
	\begin{align*}
	\begin{aligned}
	\bar \varphi_1(t,T,0) &= \frac{\sigma^2}{2} \big( \bar \varphi_3(t,T,0)\big)^2 - \kappa \theta \bar \varphi_3(t,T,0), \\
	\bar \varphi_2(t,T,0) &= \kappa \bar \varphi_3(t,T,0).  
	\end{aligned}
	\end{align*}
	We are free to specify $\varphi_3(t,T,0)$ and choose
	\begin{align*} 
	\bar \varphi_3(t,T,0) = \frac{1}{ \kappa} \Big(1- e^{-\kappa  (T-t) }\Big).
	\end{align*}
	This in turn implies that
	\begin{align*}
	\varphi_1(t,T,0) &= \frac{\sigma^2}{\kappa} \Big( e^{-\kappa  (T-t) }- e^{-2\kappa  (T-t) }\Big)-\kappa \theta   e^{-\kappa  (T-t) }, \\
	\varphi_2(t,T,0) &= \kappa  e^{-\kappa  (T-t) },\\
	\varphi_3(t,T,0) &=   e^{-\kappa  (T-t) }.
	\end{align*}
	It can be easily verified that this corresponds to the Vasi\v cek model, see Section 10.3.2.1 in \cite{Filipovic2009}. Note that this also implies  $f(t,t,0)=\xi_t$. Choosing  $r_t=f(t,t,0)$ leads to
	the num\'eraire  $X^0=\exp(\int_0^\cdot f(s,s,0) ds)$. Hence, all conditions in  Proposition \ref{prop:affine} are satisfied and  the model is free of arbitrage. An extension to the multi-curve setting is presented in Example \ref{ex:VasicekMulti}.
\end{example}

\begin{example}[A single-curve Vasi\v cek specification with discontinuity]\label{ex:VasicekJump}
	As  next step, we extend the previous example by introducing a discontinuity at time $1$. Our goal is to provide a simple, illustrative example with jump size depending on the driving process $\xi$ and we therefore remain in the single-curve framework.
	
	We assume that there is a multiplicative jump in the num\'eraire at time $T_1=1$ depending on $\exp(a\xi_1 + \epsilon)$, where $a\in\R$ and $\epsilon\sim \cN(0,b^2)$ is an independent normally distributed random variable with variance $b^2$. As driving process in \eqref{eq:fS} we consider the five-dimensional affine process
	$$ 
	X_t = \left(\int_0^t \eta(ds), \int_0^t \xi_s ds, \xi_t, \ind{t \ge 1} \xi_1,  \ind{t \ge 1} \epsilon\right)^\top, 
	$$
	where $\eta(ds)=ds+\delta_1(ds)$. The size of the jump in $X^0$ is specified by  
	$$ 
	\psi_t^\top\Delta X_t=\ind{t=1}(a \xi_1+\epsilon), 
	$$
	which can be achieved by $\psi_1^\top=(0,0,0,a,1)$. 
	The coefficients in the affine semimartingale representation \eqref{affsemchar} $\alpha_i,\ \beta_i,$ $i=0,\dots,3$, are as in Example~\ref{ex:Vasicek},  with zeros in the additional rows and columns. In addition we have $\beta_4=\beta_5=0$ and $\alpha_4=\alpha_5=0$. Moreover, 
	$$ 
	\int e^{\langle u,x \rangle}\nu^X(\{t\},dx) = \ind{t=1}\exp\left( u_1 + u_4  X_1^3 + \frac{u_5^2b^2}2 \right), \quad u \in \R^5. 
	$$
	Finally, we choose for $t\leq T$ 
	\[
	\varphi_3(t,T,0) = 
	\begin{cases}
	0 &\text{ for } t=1\leq T,\\
	a e^{-\kappa (1-t)}  & \text{ for } t< 1 = T,\\
	e^{-\kappa (T-t)}  & \text{ otherwise,}
	\end{cases}
	\]
	$\varphi_1(1,1,0) = b^2/2$, $\varphi_4(t,T,0)= (1-a)\ind{t=T=1}$, and $\varphi_5(t,T,0)=0$. $\varphi_1(t,T,0)$ for $(t,T)\neq (1,1)$ and $\varphi_2(t,T,0)$ for $t \leq T$ can be derived from $\varphi_3(t,T,0)$ as in the previous example by means of the drift condition \eqref{eq:affineI}. Condition {\em (iii)} is the interesting condition for this example. 
	This condition is equivalent to
	\begin{equation}\label{eq:affjumpcond}
	a X_1^3 - \frac{b^2}{2} = f(1-,1,0), 
	\end{equation}
	which can be satisfied by choosing $f(0,1,0)=-b^2/2$.
	Equation \eqref{eq:affjumpcond}, together with the specification of $\varphi_i(t,T,0)$ for $i=1,\dots,5$ ensures that $f(t,t,0) = \xi_t$. Choosing $r_t=f(t,t,0)$ we obtain that the model is free of arbitrage and 
	the term structure is fully specified: indeed, we recover for $1 \le t \le T$ and $0 \le t \le T <1$ the bond pricing formula from the previous example
	$$
	P(t,T,0) = \exp\Big( - A(T-t,0)- B(T-t,0) X_t^3 \Big), 
	$$
	while, for $0 \leq t < 1 \leq T$,
	\begin{align*}
	P(t,T,0) 
	=  \exp\Big(& - A(T-1,0)- A\bigl(1-t,-B(T-1,0)-a\bigr)\\
	& -B\bigl(1-t,-B(T-1,0)-a\bigr)X_t^3 +\frac{b^2}{2}\Big).	
	\end{align*}	
	The coefficients $A(\tau,u)$ and $B(\tau,u)$ are the well-known solutions of the Riccati equations, such that
	$$ 
	\E^{\Q}\left[e^{-\int_0^\tau \xi_s ds + u \xi_\tau}\right] = e^{-A(\tau,u)-B(\tau,u)\xi_0}, 
	\qquad\text{ for }\tau\geq0,
	$$
	see Section 10.3.2.1 and Corollary 10.2 in  \cite{Filipovic2009} for details and explicit formulae. 
	The example presented here extends Example 6.15 of \cite{KellerResselSchmidtWardenga2018} to a fully specified term-structure model.
\end{example}

\begin{example}[A simple multi-curve Vasi\v cek specification]\label{ex:VasicekMulti}
	We extend Example \ref{ex:Vasicek} to the multi-curve setting and consider $\cD=\{\delta\}$. For simplicity, we choose as driving diffusive part a two-dimensional Gaussian Ornstein-Uhlenbeck process:
	$$ d\xi^i_t = \kappa_i ( \theta_i - \xi^i_t) dt + \sigma_i dW^i_t, \quad i=1,\ 2,$$
	where $(W^1,W^2)^\top$ is a two-dimensional Brownian motion with correlation $\rho$. The driving process $X$ in \eqref{eq:fS} is specified as
	$$ X_t = \left(t, \int_0^t \xi^1_s ds, \xi^1_t, \int_0^t \xi^2_s ds, \xi^2 _t\right)^{\top}. $$
	The coefficients $\alpha_i$ and $\beta_i$, $i=0,\dots,5$ are time-homogeneous and obtained similarly as in Example \ref{ex:Vasicek} from \eqref{affsemchar}. Note that
	$$  \alpha_0 = \left( \begin{array}{ccccc}
	0 & 0 & 0 & 0 & 0\\ 0 & 0 & 0 & 0 & 0 \\ 0 & 0 & \sigma_1^2 & 0 & \rho\sigma_1\sigma_2 \\
	0 & 0 & 0 & 0 & 0 \\ 
	0 & 0 &  \rho\sigma_1\sigma_2 & 0 & \sigma_2^2
	\end{array}\right). $$
	The coefficients $\varphi_1(t,T,0), \dots, \varphi_3(t,T,0)$  are chosen as in Example \ref{ex:Vasicek}, while $\varphi_4(t,T,0)=\varphi_5(t,T,0)=0$.  We note that $f(t,t,0)=\xi^1_t$ and set $r_t = f(t,t,0)$. Moreover, we choose $\varphi_2(t,T,\delta)=\varphi_3(t,T,\delta)=0$ and 
	$$ \bar{\varphi}_5(t,T,\delta)  = \frac{1}{ \kappa_2} \Big(1- e^{-\kappa_2  (T-t) }\Big). $$
	Now, choose $(\psi^\delta_t)^\top=(0,1,0,-1,0)$, so that $\varphi_1(t,T,\delta)$ and $\varphi_4(t,T,\delta)$ can be calculated from ${\bar{\varphi}_5}(t,T,\delta)$ by means of the drift condition \eqref{eq:affineI}. At this stage, the model is fully specified. It is not difficult to verify that  we are in the affine framework computed in detail in Section 4.2 of \cite{BrigoMercurio01}, where explicit expressions for bond prices may be found. Moreover, we obtain $f(t,t,\delta)=\xi^2_t=X_t^5 $ and condition {\em (ii)} (and {\em (iii)}, trivially) from Proposition \ref{prop:affine} is satisfied. Condition {\em (i)} also holds: in this regard, note that
	\begin{align*}
	(\psi^\delta_t)^\top\left( \beta_0+\sum_{i=1}^5 X^i_t \beta_i \right) &= 
	(\psi^\delta_t)^\top\left( \begin{array}{c} 1 \\ X^3_t \\ \kappa_1 \theta_1 - \kappa_1 X^3_t \\ X^5_t \\ \kappa_2 \theta_2 - \kappa_2 X^5_t 		  
	\end{array}\right)    = f(t,t,0)- f(t,t,\delta).
	\end{align*} 
	Since all conditions of Proposition \ref{prop:affine} are now satisfied, we can conclude that the model is free of arbitrage. 
\end{example}

\begin{example}[A multi-curve Vasi\v cek specification with discontinuities]\label{ex:VasicekMultiJump}
	We extend the previous example by allowing for discontinuities, which can be of type I as well as of type II (see Section \ref{sec:discontinuities}) and can have a different impact on the OIS and on the Ibor curves.
	As in Example \ref{ex:VasicekMulti}, we consider a two-dimensional Gaussian Ornstein-Uhlenbeck process:
	$$ d\xi^i_t = \kappa_i ( \theta_i - \xi^i_t) dt + \sigma_i dW^i_t, \qquad i=1,\ 2.$$
	The driving process $X$ in \eqref{eq:fS} is enlarged as follows:
	$$
	X_t = \left( \int_0^t \eta(ds), \int_0^t \xi^1_s ds, \xi^1_t, \int_0^t \xi^2_s ds, \xi^2 _t, \int_0^t J_s ds, J_t \right)^\top,
	$$
	where the process $J$ is defined as
	$$
	J_t = \sum_{T_i \leq t}\epsilon_i e^{-\kappa_3 (t - T_i)}, \quad t\geq 0,
	$$
	for some $\kappa_3\geq0$. A large value of $\kappa_3$ corresponds to a high speed of mean-reversion in $J$ and generates a spiky behavior, corresponding to discontinuities of type II (recall Figure~\ref{fig:spikes}). On the contrary, a small value of $\kappa_3$ generates long-lasting jumps, which are in line with discontinuities of type I.
	For simplicity, the random variables $(\epsilon_i)_{i \ge 1}$ are i.i.d. standard normal, independent of $\xi^1$ and $\xi^2$. The set of stochastic discontinuities is described by the time points $(T_n)_{n\in\N}$ and the measure $\eta(du)$ is defined as in \eqref{eq:eta}.
	The coefficients $\alpha_i$ and $\beta_i$ are time-homogeneous and
	$$ 
	\beta_0 =  
	\left( \begin{array}{c}
	1 \\ 0 \\ \kappa_1\theta_1 \\ 0 \\ \kappa_2\theta_2\\ 0\\0
	\end{array}\right),
	\
	\beta_3 =  
	\left( \begin{array}{c}
	0 \\ 1 \\ -\kappa_1 \\ 0 \\ 0 \\ 0 \\ 0
	\end{array}\right),	
	\ 
	\beta_5 =  
	\left( \begin{array}{c}
	0 \\ 0 \\ 0 \\ 1 \\ -\kappa_2 \\0  \\ 0
	\end{array}\right),	
	\
	\beta_7 =  
	\left( \begin{array}{c}
	0 \\ 0 \\ 0 \\ 0 \\ 0 \\1\\-\kappa_3 
	\end{array}\right),	
	$$
	$\beta_1 = \beta_2 = \beta_4= \beta_6=0$,
	$$
	\alpha_0 = \left( \begin{array}{ccccccc}
	0 & 0 & 0 & 0 & 0 & 0 & 0\\ 0 & 0 & 0 & 0 & 0 & 0 & 0 \\ 0 & 0 & \sigma_1^2 & 0 & \rho\sigma_1\sigma_2 & 0 & 0\\
	0 & 0 & 0 & 0 & 0 & 0 & 0 \\ 
	0 & 0 &  \rho\sigma_1\sigma_2 & 0 & \sigma_2^2 & 0 & 0\\
	0 & 0 & 0 & 0 & 0 & 0 & 0\\
	0 & 0 & 0 & 0 & 0 & 0 & 0
	\end{array}\right),
	$$
	and $\alpha_i= 0$ for $i = 1,...,7$. Moreover, 
	\begin{equation*}
	\int_{\R^7} e^{\langle u,x \rangle} \nu^X(\{t\},dx) = \sum_{n\in\N} \ind{t = T_n} \exp\left(u_1 + \frac{u_7^2}{2}\right), \quad u \in \R^7,
	\end{equation*}
	so that 
	$$
	\gamma_0 (T_n,u)
	= u_1 + \frac{u_7^2}{2}, \quad u \in \R^7 
	$$
	and $\gamma_j(T_n,u) = 0$ for all $j = 1,\dots, 7$ and $n\in\N$. 
	
	We assume that jumps in $X^0$ and in the spread occur at the stochastic discontinuities $(T_n)_{n\in\N}$ and are specified by
	\begin{align*}
	\psi^{\top}_t \Delta X_t = \sum_{n\in \N} \ind{t = T_n} c \epsilon_n,
	&&(\psi^\delta_t)^\top \Delta X_t = \sum_{n\in \N} \ind{t = T_n} a \epsilon_n,
	\end{align*}
	which can be achieved by choosing 
	\[
	\psi_t^\top = (0,0,0,0,0,0,c)
	\qquad\text{ and }\qquad
	(\psi^\delta_t)^\top = (0,0,0,1,0,0,a).
	\]
	From this specification, it follows that the spread is given by
	$$
	S^{\delta}_t=S^{\delta}_0\exp\left(\int_0^t\xi^2_s ds+aJ_t\right).
	$$
	In line with Remark \ref{rem:disc_HJM}, the parameters $c$ and $a$ control the different impact of the stochastic discontinuities on the num\'eraire (and, hence, on the OIS curve) and on the spread (and, hence, on the Ibor curve).
	The functions $\varphi_i(t,T,0)$, for $i =1,...,7$ and $t \leq T$, are chosen as 
	\begin{align*}
	\varphi_1(t,T,0) = 
	\begin{cases}
	-\theta_1 \kappa_1 e^{-\kappa_1(T-t)} - \frac{\sigma_1^2}{\kappa_1}(e^{-2\kappa_1(T-t)}-e^{-\kappa_1(T-t)}), 
	&\text{for } t,T\notin\bbT,\\
	c e^{-\kappa_3(T-t)} -\frac{1}{\kappa_3}(e^{-2\kappa_3(T-t)}-e^{-\kappa_3(T-t)}),
	&\text{for } t\in\bbT\not\ni T,\\	
	\frac{c^2}{2},& \text{for } t=T\in\bbT,\\
	0,   &\text{otherwise},
	\end{cases}
	\end{align*}
	\begin{align*}
	\varphi_3(t,T,0) &= \varphi_3(t,T,\delta) =
	\begin{cases}
	e^{-\kappa_1 (T-t)},  & \text{for } t,T\notin\bbT,\\
	0, & \text{otherwise},
	\end{cases}\\
	\varphi_6(t,T,0) &= 
	\begin{cases}
	\kappa_3 e^{-\kappa_3 (T-t)},  & \text{for } t,T\notin\bbT,\\
	0, & \text{otherwise},
	\end{cases}\\
	\varphi_7(t,T,0) &= 
	\begin{cases}
	e^{-\kappa_3 (T-t)},  & \text{for } T\notin\bbT,\\
	0, & \text{otherwise},
	\end{cases}
	\end{align*}
	$\varphi_2(t,T,0) = \varphi_2(t,T,\delta) =
	\kappa_1 \varphi_3(t,T,0)$ and $\varphi_4(t,T,0) = \varphi_5(t,T,0) = 0$. For $\varphi(t,T,\delta)$ we choose 
	\begin{align*}
	\varphi_1(t,T,\delta) &=
	\begin{cases}
	\begin{aligned}
	&	-\theta_1 \kappa_1 e^{-\kappa_1(T-t)} - \frac{\sigma_1^2}{\kappa_1}\Big(e^{-2\kappa_1(T-t)}-e^{-\kappa_1(T-t)}\Big)\\[1mm]
	&\quad+\theta_2 \kappa_2 e^{-\kappa_2(T-t)} - \frac{\sigma_2^2}{\kappa_2}\Big(e^{-2\kappa_2(T-t)}-e^{-\kappa_2(T-t)}\Big)\\[1mm]
	&\quad +\frac{\rho \sigma_1 \sigma_2}{\kappa_1\kappa_2}
	\Big(-\kappa_1e^{-\kappa_1(T-t)}-\kappa_2e^{-\kappa_2(T-t)}\\
	& \hphantom{\quad+\frac{\rho \sigma_1 \sigma_2}{\kappa_1\kappa_2}
		\Big(\big)}+(\kappa_1+\kappa_2)e^{-(\kappa_1+\kappa_2)(T-t)}\Big),
	\end{aligned}	
	&\text{for } t,T\notin\mathbb{T},\\
	&\\
	\begin{aligned}
	&\frac{(1+a\kappa_3)}{\kappa_3}\Big((1+c\kappa_3)e^{-\kappa_3(T -t)}-(1+a\kappa_3)e^{-2\kappa_3(T -t)}\Big),
	& 	 	\end{aligned}
	& 
	\text{for } t\in\mathbb{T}\not\ni T,\\
	&\\
	\half(a-c)^2,&  \text{for } t=T\in \mathbb{T},\\
	&\\
	0,   &\text{otherwise},
	\end{cases}
	\end{align*}
	\begin{align*}
	\varphi_5(t,T,\delta) &= 
	\begin{cases}
	-e^{-\kappa_2 (T-t)},  & \text{for } t,T\notin\bbT,\\
	0, & \text{otherwise},
	\end{cases}\\[2mm]
	\varphi_6(t,T,\delta) &= 
	\begin{cases}
	\kappa_3 (1+a\kappa_3)e^{-\kappa_3 (T-t)},  & \text{for } t,T\notin\bbT,\\
	0, & \text{otherwise},
	\end{cases}\\[2mm]
	\varphi_7(t,T,\delta) &= 
	\begin{cases}
	(1+a\kappa_3)e^{-\kappa_3 (T-t)},  & \text{for } T\notin\bbT,\\
	0, & \text{otherwise},
	\end{cases}
	\end{align*}
	and $\varphi_4(t,T,\delta) = \kappa_2 \varphi_5(t,T,\delta)$. With this specification, it can be checked that condition {\em (ii)} of Proposition \ref{prop:affine} is satisfied. Furthermore, it can be verified that 
	\[
	f(t,t,0) = \xi^1_t + J_t
	\qquad\text{ and }\qquad
	f(t,t,\delta) = \xi^1_t-\xi^2_t + (1+a\kappa_3)J_t.
	\]
	Therefore, condition {\em (i)} of Proposition \ref{prop:affine} is satisfied by setting $r_t = \xi^1_t + J_t$.
	By choosing $f(0,T_n,0) = -c^2/2$ and $f(0,T_n,\delta) = -\half(a-c)^2$ and calculating 
	\begin{align*}
	\int_{(T_n,T]}\varphi_1(T_n,u,0)\eta(du) 
	&= -\frac{c}{\kappa_3}\big(e^{-\kappa_3 (T-T_n)}-1\big) + \frac{1}{2\kappa_3^2}\big(e^{-\kappa_3 (T-T_n)}-1\big)^2,\\[1mm]
	\int_{(T_n,T]}\varphi_7(T_n,u,0)\eta(du) 
	&=-\frac{1}{\kappa_3}\big(e^{-\kappa_3 (T-T_n)}-1\big),\\[1mm]
	\int_{(T_n,T]}\varphi_1(T_n,u,\delta)\eta(du) 
	&= \frac{(a-c)(1+a\kappa_3)}{\kappa_3}\big(e^{-\kappa_3 (T-T_n)}-1\big) \\
	&\quad+ \frac{(1+a\kappa_3)^2}{2\kappa_3^2}\big(e^{-\kappa_3 (T-T_n)}-1\big)^2,\\[1mm]
	\int_{(T_n,T]}\varphi_7(T_n,u,\delta)\eta(du) 
	&=-\frac{(1+a\kappa_3)}{\kappa_3}\big(e^{-\kappa_3 (T-T_n)}-1\big),
	\end{align*}
	we can see that condition {\em (iii)}
	\begin{eqnarray*}
		- f(T_n-,T_n,0)
		&=& -\int_{(T_n,T]}\varphi_1(T_n,u,0)\eta(du) + \half \left(-c -\int_{(T_n,T]}\varphi_7(T_n,u,0)\eta(du)\right)^2,\\
		- f(T_n-,T_n,\delta)
		&=&-\int_{(T_n,T]}\varphi_1(T_n,u,\delta)\eta(du)+\half \left(a-c-\int_{(T_n,T]}\varphi_7(T_n,u,\delta)\eta(du)\right)^2,		
	\end{eqnarray*}
	is satisfied for all $n\in\N$ and $T \geq T_n$. We can conclude that the term structure is fully specified and, by Proposition \ref{prop:affine}, the model is free of arbitrage.
\end{example}

\section{An FTAP for multiple curve financial markets} \label{sec:NAFLVR}

In this section, we characterize absence of arbitrage in a multiple curve financial market. At the present level of generality, this represents the first rigorous analysis of absence of arbitrage in post-crisis fixed-income markets.

As introduced in Definition \ref{def:market}, a multiple curve financial market is a {\em large financial market} containing uncountably many securities. An economically convincing notion of no-arbitrage for large financial markets has been introduced in \cite{CuchieroKleinTeichmann} under the name of {\em no asymptotic free lunch with vanishing risk} (NAFLVR), generalizing the classic requirement of NFLVR for finite-dimensional markets (see \cite{DelbaenSchachermayer1994} and \cite{CuchieroTeichmann:Emery}). In this section, we extend the main result of \cite{CuchieroKleinTeichmann} to an infinite time horizon and apply it to a general multiple curve financial market.

Let $(\Omega, \cF, \bbF, \P)$ be a filtered probability space satisfying the usual conditions of right-continuity and $\P$-completeness, with $\cF:=\bigvee_{t\geq0}\cF_t$.
Let us recall that a process $Z=(Z_t)_{t\geq0}$ is said to be a {\em semimartingale up to infinity} if there exists a process $\overline{Z}=(\overline{Z}_t)_{t\in[0,1]}$ satisfying $\overline{Z}_t = Z_{t/(1-t)}$, for all $t<1$, and such that $\overline{Z}$ is a semimartingale with respect to the filtration $\overline{\bbF}=(\overline{\ccF}_t)_{t\in[0,1]}$ defined by
\[
\overline{\ccF}_t = \begin{cases}
\ccF_{\frac{t}{1-t}}, &\text{ for }t<1,\\
\ccF, &\text{ for }t=1,
\end{cases}
\]
see Definition 2.1 in \cite{CS05}.
We denote by $\bbS$ the space of real-valued semimartingales up to infinity equipped with the Emery topology, see \cite{Stricker81}.
For a set $C\subset \bbS$, we denote by $\overline{C}^{\,\bbS}$ its closure with respect to the Emery topology.

We denote by $\cI:=\R_+ \times  \cD_0\times \R$ the parameter space characterizing the traded assets included in Definition \ref{def:market}. We furthermore assume the existence of a tradable num\'eraire with strictly positive adapted price process $X^0$.
For notational convenience, we represent OIS zero-coupon bonds by setting $\pifra(t,T,0,K):=P(t\wedge T,T)$, for all $(t,T)\in\R^2_+$ and $K\in\R$. 
We also set $\pifra(t,T,\delta,K)=\pifra(T+\delta,T,\delta,K)$ for all  $\delta\in\cD$, $K\in\R$ and $t\geq T+\delta$.

For $n\in\N$, we denote by $\cI^n$ the family of all subsets $A\subset\cI$ containing $n$ elements. For each  $A=((T_1,\delta_1,K_1),\dots,(T_n,\delta_n,K_n)) \in \cI^n$, we define the collection of $X^0$-discounted prices $\bS^A=(S^1,\dots,S^n)$ by 
\[ 
S^i := (X^0)^{-1}\pifra(\cdot,T_i,\delta_i,K_i), 
\qquad\text{ for } i=1,\ldots,n.
\]
For each $A\in\cI^n$, $n\in\N$, we assume that $\bS^A$ is a semimartingale on $(\Omega,\bbF,\P)$ and we denote by $L_{\infty}(\bS^A)$ the set of all $\R^{|A|}$-valued, predictable processes $\btheta=(\theta^1,\ldots,\theta^{|A|})$ which are integrable up to infinity with respect to $\bS^A$, in the sense of Definition 4.1 in \cite{CS05}.
We assume that trading occurs in a self-financing way and we say that a process $\btheta\in L_{\infty}(\bS^A)$ is a {\em $1$-admissible trading strategy} if $\btheta_0=0$ and $(\btheta\cdot\bS^A)_t\geq-1$ a.s. for all $t\geq0$. The set $\cX^A_1$ of wealth processes generated by $1$-admissible trading strategies with respect to $\bS^A$ is defined as
\begin{equation*}
\cX^A_1 
:= \bigl\{ \btheta \cdot \bS^A: \btheta\in L_{\infty}(\bS^A) \text{ and } \btheta \text{ is $1$-admissible}\bigr\} 
\subset \bbS.
\end{equation*}
The set of wealth processes generated by trading in at most $n$ arbitrary assets is given by $ \cX^n_1 = \medcup_{A \in \cI^n} \cX^A_1 $. 
By allowing to trade in arbitrary finitely many assets and letting the number of assets increase to infinity, we arrive at generalized portfolio wealth processes. 
The corresponding set of $1$-admissible wealth processes is given by $ \cX_1 := \overline{\medcup_{n \in \mathbb{N}} \cX^n_1}^{\,\bbS}$, so that all admissible generalized portfolio wealth processes in the multiple curve financial market are finally given by
\[ 
\cX := \bigcup_{\lambda >0} \lambda \cX_1.
\]

\begin{remark}	\label{rem:single_strike}
	The set $\cX$ can be equivalently described as the set of all admissible generalized portfolio wealth processes which can be constructed in the financial market consisting of the following two subsets of assets:
	\begin{enumerate}
		\item[(i)] 
		OIS zero-coupon bonds, for all maturities $T\in\R_+$,
		\item[(ii)] 
		FRAs, for all tenors $\delta\in\cD$, all settlement dates $T\in\R_+$ and strike $K'$,
	\end{enumerate}
	for some {\em fixed arbitrary} strike $K'\in\R$.
	This follows from our standing assumption of linear valuation of FRAs together with the associativity of the stochastic integral.
\end{remark}

Since each element $X\in\cX$ is a semimartingale up to infinity,  the limit $X_{\infty}$ exists pathwise and is finite.
We can therefore define $K_0:=\{X_{\infty}: X \in \cX\}$ 
and $C:=(K_0-L^0_+)\medcap L^\infty$, the convex cone of bounded claims super-replicable with zero initial capital.

\begin{definition}  \label{def:NAFLVR}
	We say that the multiple curve financial market satisfies no asymptotic free lunch with vanishing risk (NAFLVR) if 
	\[ 
	\overline C \medcap L^\infty_+ = \{0\},
	\]
	where $\overline C$ denotes the norm closure in $L^\infty$ of the set $C$.
\end{definition}

The following result provides a general formulation of the fundamental theorem of asset pricing for multiple curve financial markets. 

\begin{theorem}	\label{thm:FTAP}
	The multiple curve financial market satisfies NAFLVR if and only if there exists an equivalent separating measure $\Q$, i.e., a probability measure $\Q\sim \P$ on $(\Omega,\cF)$ such that $\E^\Q[X_{\infty}] \le 0$ for all $X \in \cX$.
\end{theorem}

\begin{proof}
	We divide the proof into several steps, with the goal of reducing our general multiple curve financial market to the setting considered in \cite{CuchieroKleinTeichmann}.
	
	1) In view of Remark \ref{rem:single_strike}, it suffices to consider FRA contracts with fixed strike $K=0$, for all tenors $\delta\in\cD$ and settlement dates $T\in\R_+$. Consequently, the parameter space $\cI=\R_+\times\cD_0\times\R$ can be reduced to $\cI':=\R_+\times\{0,1,\ldots,m\}$, which can be further transformed into a subset of $\R_+$ via
	$\cI'\owns(T,i)\mapsto i+T/(1+T)\in[0,m+1)=:\mathcal{J}$.
	
	2) Without loss of generality, we can assume that $(X^0)^{-1}\pifra(\cdot,T,\delta,0)$ is a semimartingale up to infinity, for every $T\in\R_+$ and $\delta\in\cD_0$.
	Indeed, let $n\in\N$ and $A\in\mathcal{J}^n$. Similarly as in the proof of \cite[Theorem 5.5]{CS05}, 
	for each $i=1,\ldots,n$, there exists a deterministic function $K^i>0$ such that $(K^i)^{-1}\in L(S^i)$ and $Y^i:=(K^i)^{-1}\cdot S^i\in\bbS$. 
	Setting ${\bm Y}^A=(Y^1,\ldots,Y^n)$, the associativity of the stochastic integral together with \cite[Theorem 4.2]{CS05} allows to prove that
	\[
	\cX^A_1 = \bigl\{\bphi\cdot {\bm Y}^A : \bphi\in L_{\infty}({\bm Y}^A), \bphi_0=0 \text{ and }(\bphi\cdot{\bm Y}^A)_t\geq -1\text{ a.s. for all }t\geq0\bigr\}.
	\]
	Henceforth, we shall assume that $\bS^A\in\bbS$, for all $A\in\mathcal{J}^n$ and $n\in\N$.
	
	3) For $t\in[0,1)$ and $u\in[0,+\infty)$, let $\alpha(t):=t/(1-t)$ and $\beta(u):=u/(1+u)$. The functions $\alpha$ and $\beta$ are two inverse isomorphisms between $[0,1)$ and $[0,+\infty)$ and can be extended to $[0,1]$ and $[0,+\infty]$.
	For $A\in\mathcal{J}^n$, $n\in\N$, let us define the process $\overline{\bS}^A=(\overline{\bS}^A_t)_{t\in[0,1]}$ by $\overline{\bS}^A_t:=\bS^A_{\alpha(t)}$, for all $t\in[0,1]$. Since $\bS^A\in\bbS$, the process $\overline{\bS}^A$ is a semimartingale on $(\Omega,\overline{\bbF},\P)$.
	Let $\btheta\in L_{\infty}(\bS^A)$. We define the process $\overline{\btheta}=(\overline{\btheta}_t)_{t\in[0,1]}$ by $\overline{\btheta}_t:=\btheta_{\alpha(t)}$, for all $t<1$, and $\overline{\btheta}_1:=0$. As in the proof of \cite[Theorem 4.2]{CS05}, it holds that $\overline{\btheta}\in L(\overline{\bS}^A)$.
	Moreover, it can be shown that 
	\beq	\label{eq:bar_processes}
	(\overline{\btheta}\cdot\overline{\bS}^A)_t=(\btheta\cdot\bS^A)_{\alpha(t)},
	\eeq
	for all $t\in[0,1]$.
	Conversely, if $\overline{\btheta}\in L(\overline{\bS}^A)$, the process $\btheta=(\btheta_t)_{t\geq0}$ defined by $\btheta_t:=\overline{\btheta}_{\beta(t)}$, for $t\geq0$, belongs to $L_{\infty}(\bS^A)$ and it holds that
	\begin{equation*}
	(\btheta\cdot{\bm S}^A)_t=(\overline{\btheta}\cdot\overline{{\bm S}}^A)_{\beta(t)},
	\end{equation*}
	for all $t\geq0$.
	Furthermore, $(\btheta\cdot{\bm S}^A)_{\infty}=(\overline{\btheta}\cdot\overline{{\bm S}}^A)_1$ holds if $\overline{\btheta}_1=0$.
	
	4) In view of step 3), we can consider an equivalent financial market indexed over $[0,1]$ in the filtration $\overline{\bbF}$. To this effect, for each $A\in\mathcal{J}^n$, $n\in\N$, let us define
	\[
	\overline{\cX}^A_1 := \bigl\{\overline{\btheta}\cdot\overline{{\bm S}}^A : \overline{\btheta}\in L(\overline{{\bm S}}^A), \, 
	\overline{\btheta}_0=\overline{\btheta}_1=0 
	\text{ and }
	(\overline{\btheta}\cdot\overline{{\bm S}}^A)_t\geq -1\text{ a.s. } \forall t\in[0,1]\bigr\}
	\] 
	and the sets
	\[
	\overline{\cX}^n_1:=\bigcup_{A\in\mathcal{I}^n}\overline{\cX}^A_1,
	\qquad
	\overline{\cX}_1:=\overline{\bigcup_{n\in\N}\overline{\cX}^n_1}^{\,\bbS},
	\qquad
	\overline{\cX} := \bigcup_{\lambda>0}\lambda\overline{\cX}_1
	\]
	and $\overline{K}_0:=\{\overline{X}_1 : \overline{X}\in\overline{\cX}\}$, where the closure in the definition of $\overline{\cX}_1$ is taken in the semimartingale topology on the filtration $\overline{\bbF}$.
	Let $(X^k)_{k\in\N}\subseteq\medcup_{n\in\N}\cX^n_1$ be a sequence converging to $X$ in the topology of $\bbS$ (on the filtration $\bbF$). By definition, for each $k\in\N$, there exists a set $A_k$ such that $X^k=\btheta^k\cdot{\bm S}^{A_k}$ for some $1$-admissible strategy $\btheta^k\in L_{\infty}({\bm S}^{A_k})$. 
	In view of \eqref{eq:bar_processes}, it holds that
	$$ 
	X^k_{\alpha(t)}=(\overline{\btheta}^k\cdot\overline{{\bm S}}^{A_k})_t=:\overline{X}^k_t, $$ 
	for all $t\in[0,1]$.
	Since the topology of $\bbS$ is stable with respect to changes of time (see Proposition 1.3 in \cite{Stricker81}), the sequence $(\overline{X}^k)_{k\in\N}$ converges in the semimartingale topology (on the filtration $\overline{\bbF}$) to $\overline{X}=X_{\alpha(\cdot)}\in\overline{\cX}_1$. This implies that $K_0 \subseteq \overline{K}_0$.
	An analogous argument allows to show the converse inclusion, thus proving that 
	$
	K_0 = \overline{K}_0.
	$
	In view of Definition \ref{def:NAFLVR}, this implies that NAFLVR holds for the original financial market if and only if it holds for the equivalent financial market indexed over $[0,1]$ on the filtration $\overline{\bbF}$.
	
	5) It remains to show that, for every $A\in\mathcal{J}^n$, $n\in\N$, the set $\overline{\cX}^A_1$ satisfies the requirements of \cite[Definition 2.1]{CuchieroKleinTeichmann}.
	First, $\overline{\cX}^A_1$ is convex and, by definition, each element $\overline{X}\in\overline{\cX}^A_1$ starts at $0$ and is uniformly bounded from below by $-1$.
	Second, let $\overline{X}^1,\overline{X}^2 \in \overline{\cX}^A_1$ and two bounded $\overline{\bbF}$-predictable processes $H^1,H^2 \ge 0$ such that $H^1 H^2=0$. By definition, there exist processes $\overline{\btheta}^1$ and $\overline{\btheta}^2$ such that $\overline{X}^i=\overline{\btheta}^i \cdot \overline{\bS}^A$, for $i=1,2$. If $Z:=H^1 \cdot \overline{X}^1+H^2\cdot \overline{X} ^2 \ge -1$, then 
	$$ 
	Z =(H^1 \overline{\btheta}^1 + H^2 \overline{\btheta}^2) \cdot \overline{\bS}^A\in \overline{\cX}^A_1,
	$$  
	so that the required concatenation property holds. Moreover, $\overline{\cX}^{A^1} \subset \overline{\cX}^{A^2}$ if $A^1 \subset A^2$.
	The theorem finally follows from \cite[Theorem 3.2]{CuchieroKleinTeichmann}.
\end{proof}

\begin{remark}   \label{rem:ELMM}
	An {\em equivalent local martingale measure} (ELMM) is a probability measure $\Q\sim\P$ on $(\Omega,\cF)$ such that $(X^0)^{-1}\pifra(\cdot,T,\delta,K)$ is a $\Q$-local martingale, for all $T\in\R_+$, $\delta\in\cD_0$ and $K\in\R$.
	Under additional conditions (namely of locally bounded discounted price processes, see \cite[Section 3.3]{CuchieroKleinTeichmann}), it can be shown that NAFLVR is equivalent to the existence of an ELMM. In general, one cannot replace in Theorem \ref{thm:FTAP} a separating measure with an ELMM, as shown by an explicit counterexample in \cite{CuchieroKleinTeichmann}. However, as a consequence of Fatou's lemma, the existence of an ELMM always represents a sufficient condition for NAFLVR. Assuming that the num\'eraire $X^0$ is tradable, an ELMM corresponds to a {\em risk-neutral measure} (see Section \ref{sec:TSM}), which has been precisely characterized in the previous sections of the paper.
\end{remark}

\begin{remark} 
	Absence of arbitrage in large financial markets has also been studied by \cite{KK98} in the sense of {\em no asymptotic arbitrage of the first kind} (NAA1), which is a weaker requirement than NAFLVR, see \cite[Section 4]{CuchieroKleinTeichmann}. Differently from \cite{KK98}, we work on a fixed filtered probability space $(\Omega,\cF,\bbF,\mathbb{P})$ and not on a sequence of probability spaces. On the other hand, we allow for uncountably many traded assets (see Definition \ref{def:market}).
\end{remark}

\section{Conclusions}

The aim of this paper has been to introduce stochastic discontinuities into term structure modeling in a multi-curve setup.  Stochastic discontinuities are a key feature in interest rate markets and we introduced two types for the classification of these jumps. To this end, we provided a general analysis of post-crisis multiple curve markets under minimal assumptions. 
	
	Three key results have been developed in our work: first, we provide a characterization of absence of arbitrage in an extended HJM setting. Second, we provide a similar characterization for market models. Both results rely on a fundamental theorem of asset pricing for multiple curve financial markets. Third, we provide a flexible class of multi-curve models based on affine semimartingales, a setup  allowing for stochastic discontinuities. 
	
	While the focus of our analysis is a fundamental treatment of pricing in multiple curve markets, it is worth emphasizing that this framework has a large potential for many other applications such as risk management, requiring further studies. In particular for the latter, a proper modeling  of the market price of risk and taking macro-economic variables into account are equally important.

\appendix
\section{Technical results}\label{app:proofs}

The following technical result on ratios and products of stochastic exponentials easily follows from Yor's formula, see \cite[\textsection~II.8.19]{JacodShiryaev}.

\begin{corollary} \label{cor:stoch_exp_property}For any semimartingales $X$, $Y$ and $Z$ with $\Delta Z > -1$, it holds that
	\begin{align*}
	\frac{\mathcal{E}(X)\mathcal{E}(Y)}{\mathcal{E}(Z)} = \mathcal{E}\Bigg(  & X + Y - Z + \langle X^{\textup{c}},Y^{\textup{c}}\rangle- \langle Y^{\textup{c}},Z^{\textup{c}}\rangle- \langle X^{\textup{c}},Z^{\textup{c}}\rangle +\langle Z^{\textup{c}},Z^{\textup{c}}\rangle \\
	&  +\sum_{0<s\leq \cdot} \bigg(\frac{\Delta Z_s(-\Delta X_s - \Delta Y_s + \Delta Z_s) + \Delta X_s \Delta Y_s}{1+\Delta Z_s}\bigg)\Bigg).
	\end{align*}
\end{corollary}

\begin{proof}\textit{of Lemma \ref{lem:bond_prices}}
	Due to Assumption \ref{ass} it can be verified by means of Minkowski's integral inequality and H\"older's inequality that the stochastic integrals appearing in \eqref{eq:dec_PtT} are well-defined, for every $T\in\R_+$ and $\delta\in\cD_0$.
	Let $F(t, T,\delta) : = \int_{(t, T]}f(t, u, \delta)\eta(du)$, for all $0\leq t\leq T<+\infty$. For $t<T$, equation \eqref{eq:fwd_rate} implies that 
	\begin{align*}
	F(t,T,\delta)
	& = \int_{(t, T]}\left( f(0, u, \delta) + \int_0^t a(s,u,\delta) ds 
	+V(t,u,\delta) + \int_0^t b(s,u,\delta) dW_s \right.\\
	& \left.   
	+ \int_0^t \int_E g(s, x, u,\delta) \bigl(\mu(ds,dx) - \nu(ds,dx)\bigr) \right)\eta(du)\\
	& = \int_0^T f(0, u, \delta)\eta(du) + \int_0^T \int_0^t a(s, u, \delta) ds\eta(du)+ \int_0^T V(t,u,\delta) \eta(du) \\
	&  
	+ \int_0^T \int_0^t b(s, u, \delta) dW_s\eta(du) + \int_0^T \int_0^t \int_E g(s, x, u,\delta) \bigl(\mu(ds,dx) - \nu(ds,dx)\bigr) \eta(du) \\
	&  - \int_0^t f(0, u, \delta)\eta(du) 
	- \int_0^t \int_0^u a(s, u, \delta) ds\eta(du) - \int_0^tV(u,u,\delta) \eta(du)  \\
	&  - \int_0^t \int_0^u b(s, u, \delta) dW_s\eta(du) - \int_0^t \int_0^u \int_E g(s, x, u,\delta) \bigl(\mu(ds,dx) - \nu(ds,dx)\bigr) \eta(du).
	\end{align*}
	Due to Assumption \ref{ass}, we can apply ordinary and stochastic Fubini theorems, in the versions of Theorem 2.2 in \cite{Veraar12} for the stochastic integral with respect to $W$ and in the version of Proposition A.2 in \cite{Bjoerk1997} for the stochastic integral with respect to the compensated random measure $\mu-\nu$. We therefore obtain
	\begin{align} \label{eq:485}
	F(t,T,\delta)
	& = \int_0^T f(0, u, \delta)\eta(du)  - \int_0^t f(u, u, \delta)\eta(du) + \int_0^t \int_{[s,T]}\! a(s, u, \delta) \eta(du)ds \nonumber\\
	&\quad + \int_0^TV(t,u,\delta)\eta(du) 
	+ \int_0^t \int_{[s,T]} b(s, u, \delta) \eta(du) dW_s \nonumber\\
	& \quad + \int_0^t \int_E \int_{[s,T]} g(s, x, u,\delta) \eta(du) \bigl(\mu(ds,dx) - \nu(ds,dx)\bigr)  \nonumber\\
	& = \int_0^T f(0, u, \delta)\eta(du) + \int_0^t \bar{a}(s, T, \delta) ds + \sum_{n\in\N} \bar{V}(T_n, T, \delta)\ind{T_n \leq t} + \int_0^t \bar{b}(s, T, \delta) dW_s\nonumber\\
	&  \quad
	 	+ \int_0^t \int_E \bar{g}(s, x, T,\delta) \bigl(\mu(ds,dx) - \nu(ds,dx)\bigr) - \int_0^t f(u, u, \delta)\eta(du)\nonumber\\
	&=: G(t,T,\delta).
	\end{align}
	In \eqref{eq:485}, the finiteness of $\int_0^{\cdot}f(u,u,\delta)\eta(du)$ follows by Assumption \ref{ass} together with an analogous application of ordinary and stochastic Fubini theorems.
	
	To complete the proof, it remains to establish \eqref{eq:dec_PtT} for $t = T\in\R_+$. To this effect, it suffices to show that $\Delta G(T,T,\delta)=\Delta F(T,T,\delta)$ for all $T\in\R_+$, where $\Delta G(T,T,\delta):=G(T,T,\delta)-G(T-,T,\delta)$, and similarly for $\Delta F(T,T,\delta)$.  
	By \cite[Proposition II.1.17]{JacodShiryaev},  $\nu(\{T\}\times E)=0$ implies that, for every $T\in\R_+$, $\Q[\mu(\{T\}\times E)\neq0]=0$. Therefore, it holds that $\Q[\Delta G(T,T,\delta)\neq0]>0$ only if $T=T_n$, for some $n\in\N$.
	For $T=T_1$, equations \eqref{eq:485} and \eqref{eq:fwd_rate} together imply that
	\begin{align*}
	\Delta G(T_1,T_1,\delta)
	&= \bar{V}(T_1,T_1,\delta) - f(T_1,T_1,\delta) 
	= -f(T_1-,T_1,\delta)= -F(T_1-,T_1,\delta)	\\
	& = \Delta F(T_1,T_1,\delta),
	\end{align*}
	where the last equality follows from the convention $F(T_1,T_1,\delta)=0$. By induction over $n$, the same reasoning yields that 
	$$
	\Delta G(T_n,T_n,\delta)=\Delta F(T_n,T_n,\delta),
	$$ for all $n\in\N$.
	Finally, the semimartingale property of $\delta$-tenor bond prices $(P(t,T,\delta))_{0\leq t\leq T}$ follows from \eqref{eq:485}.
\end{proof}

\section{Embedding of market models into the HJM framework} \label{app:MM}

The general market model considered in Section \ref{sec:market models}, as specified by equation \eqref{def:Ldyn}, can be embedded into the extended HJM framework of Section \ref{sec:TSM}. For simplicity of presentation, let us consider a market model for a single tenor (i.e., $\cD=\{\delta\}$) and suppose that the forward Ibor rate $L(\cdot,T,\delta)$ is given by \eqref{def:Ldyn}, for all $T\in\cT^{\delta}=\{T_1,\ldots,T_N\}$, with $T_{i+1}-T_i=\delta$ for all $i=1,\ldots,N-1$.
Always for simplicity, let us assume that there is a fixed number $N+1$ of discontinuity dates, coinciding with the set of dates $\cT^0:=\cT^{\delta}\medcup\{T_{N+1}\}$, with $T_{N+1}:=T_N+\delta$.
We say that $\{L(\cdot,T,\delta):T\in\cT^{\delta}\}$ can be {\em embedded into an extended HJM model} if there exists a sigma-finite measure $\eta$ on $\R_+$, a spread process $S^{\delta}$ and a family of forward rates $\{f(\cdot,T,\delta):T\in\cT^{\delta}\}$ such that
\beq	\label{eq:embed}
L(t,T,\delta) = \frac{1}{\delta}\left(S^{\delta}_t\frac{P(t,T,\delta)}{P(t,T+\delta)}-1\right),
\qquad\text{ for all }0\leq t\leq T\in \cT^{\delta},
\eeq
where $P(t,T,\delta)$ is given by \eqref{eq:PtT}, for all $0\leq t\leq T\in\cT^{\delta}$. In other words, in view of equation \eqref{eq:FRA}, the HJM model generates the same forward Ibor rates as the original market model, for every date $T\in\cT^{\delta}$.

We remark that, since a market model involves OIS bonds only for maturities $\cT^0=\{T_1,\ldots,T_{N+1}\}$, there is no loss of generality in taking the measure $\eta$ in \eqref{eq:PtT} as a purely atomic measure of the form
\beq	\label{eq:meas_embed}
\eta(du) = \sum_{i=1}^{N+1}\delta_{T_i}(du).
\eeq
More specifically, if OIS bonds for maturities $\cT^0$ are defined through \eqref{eq:PtT} via a generic measure of the form \eqref{eq:eta}, then there always exists a measure $\eta$ as in \eqref{eq:meas_embed} generating the same bond prices, up to a suitable specification of the forward rate process.

The following proposition explicitly shows how a general market model  can be embedded into an HJM model. 
For $t\in[0,T_N]$, we define 
$$
i(t):=\min\{j\in\{1,\ldots,N\}:T_j\geq t\},
$$
so that $T_{i(t)}$ is the smallest  $T\in\cT^{\delta}$ such that $T\geq t$.

\begin{proposition}	\label{prop:embed}
	Suppose that all the conditions of Theorem \ref{thm:Libor} are satisfied, with respect to the measure $\eta$ given in \eqref{eq:meas_embed}, and assume furthermore that $L(t,T,\delta)>-1/\delta$ a.s. for all $t\in[0,T]$ and $T\in\cT^{\delta}$. 
	Then, under the above assumptions, the market model $\{L(\cdot,T,\delta):T\in\cT^{\delta}\}$ can be embedded into an HJM model by choosing
	\begin{enumerate}
		\item[(i)]
		a family of forward rates $\{f(\cdot,T,\delta):T\in\cT^{\delta}\}$ with initial values
		\[
		f(0,T_i,\delta) = f(0,T_{i+1},0) - \log\left(\frac{1+\delta L(0,T_i,\delta)}{1+\delta L(0,T_{i-1},\delta)}\right),
		\qquad\text{ for }i=1,\ldots,N,
		\]
		and satisfying equation \eqref{eq:fwd_rate} where, for all $i=1,\ldots,N$, the volatility process $b(\cdot,T_i,\delta)$, the jump function $g(\cdot,\cdot,T_i,\delta)$ and $(\Delta V(T_n,T_i,\delta))_{n=1,\ldots,N}$ are respectively given by
		\[
		b(t,T_i,\delta) = 
		\begin{cases}
		b(t,T_i,0) + b(t,T_{i+1},0) - \delta\frac{b^L(t,T_i,\delta)}{1+\delta L(t-,T_i,\delta)}, 
		&\text{if }i=i(t),\\
		b(t,T_{i+1},0) - \delta\left(\frac{b^L(t,T_i,\delta)}{1+\delta L(t-,T_i,\delta)}-\frac{b^L(t,T_{i-1},\delta)}{1+\delta L(t-,T_{i-1},\delta)}\right), 
		&\text{if }i>i(t),
		\end{cases}
		\]
		\[
		g(t,x,T_i,\delta) = 
		\begin{cases}
		\begin{aligned}
		&g(t,x,T_{i+1},0) + g(t,x,T_i,0)  \\
		&- \log\left(1+\frac{\delta g^L(t,x,T_i,\delta)}{1+\delta L(t-,T_i,\delta)}\right),
		\end{aligned}
		& \text{ if }i=i(t),\\
		&\\
		g(t,x,T_{i+1},0) 
		- \log\left(\frac{1+\frac{\delta g^L(t,x,T_i,\delta)}{1+\delta L(t-,T_i,\delta)}}{1+\frac{\delta g^L(t,x,T_{i-1},\delta)}{1+\delta L(t-,T_{i-1},\delta)}}\right),
		& \text{ if }i>i(t),
		\end{cases}	
		\]
		\[
		\Delta V(T_n,T_i,\delta) = \Delta V(T_n,T_{i+1},0) - \log\left(\frac{\frac{1+\delta L(T_n,T_i,\delta)}{1+\delta L(T_n-,T_i,\delta)}}{\frac{1+ \delta L(T_n,T_{i-1},\delta)}{1+\delta L(T_n-,T_{i-1},\delta)}}\right),
		\quad\text{ for }i\geq n+1,
		\]
		and the process $a(\cdot,T_i,\delta)$ is determined by condition (ii) of Theorem \ref{thm:HJM};
		\item[(ii)]
		a spread process $S^{\delta}$ with initial value
		$
		S^{\delta}_0 = \bigl(1+\delta L(0,0,\delta)\bigr)P(0,\delta)
		$
		and satisfying \eqref{eq:spread}, \eqref{eq:spread_FV}, where the processes $\alpha^{\delta}$, $H^{\delta}$, the function $L^{\delta}$ and the random variables $(\Delta A^{\delta}_{T_n})_{n=1,\ldots,N}$ are respectively given by
		\begin{align*}
		\alpha^{\delta}_t &= 0,	\qquad\qquad\qquad
		H^{\delta}_t = 0,	\qquad\qquad\qquad
		L^{\delta}(t,x) = 0,\\
		\Delta A^{\delta}_{T_n} &= \left(\frac{1+\delta L(T_n,T_n,\delta)}{1+\delta L(T_n-,T_n,\delta)}\right)e^{f(T_n-,T_n,0)-f(T_n-,T_n,\delta)-\Delta V(T_n,T_{n+1},0)}-1.
		\end{align*}
	\end{enumerate}
	Moreover, the resulting HJM model satisfies all the conditions of Theorem \ref{thm:HJM}.
\end{proposition}
\begin{proof}
	Since the proof involves rather lengthy computations, we shall only provide a sketch. 
	For $T\in\cT^{\delta}$, by means of Theorem \ref{thm:Libor} and the assumption $L(t,T,\delta)>-1/\delta$ a.s. for all $t\in[0,T]$, the process $(1+\delta L(\cdot,T,\delta))P(\cdot,T+\delta)/X^0$ is a strictly positive $\Q$-local martingale, so that $L(t-,T,\delta)>-1/\delta$  a.s. for all $t\in[0,T]$ and $T\in\cT^{\delta}$. 
	Let us define the process $Y(T,\delta)=(Y_t(T,\delta))_{0\leq t\leq T}$ by $Y_t(T,\delta):=S^{\delta}_tP(t,T,\delta)/P(t,T+\delta)$. An application of Corollary \ref{cor:stoch_exp_property}, together with equation \eqref{eq:spread} and Corollary \ref{cor:bond_stoch_exp}, 
	yields a stochastic exponential representation and a semimartingale decomposition of the process $Y(T,\delta)$. 
	
	For the spread process $S^{\delta}$ given in \eqref{eq:spread}, we start by imposing $H^{\delta}=0$ and $L^{\delta}=0$. 
	We then proceed to determine the processes describing the forward rates $\{f(\cdot,T,\delta):T\in\cT^{\delta}\}$ satisfying \eqref{eq:fwd_rate}.
	In view of \eqref{eq:embed}, for each $T\in\cT^{\delta}$, we determine the process $b(\cdot,T,\delta)$ by matching the Brownian part of  $Y(T,\delta)$ with the Brownian part of $\delta L(\cdot,T,\delta)$, while the jump function $g(\cdot,\cdot,T,\delta)$ is obtained in a similar way by matching the totally inaccessible jumps of $Y(T,\delta)$ with the totally inaccessible jumps of $\delta L(\cdot,T,\delta)$. 
	The drift process $a(\cdot,T,\delta)$ is then univocally determined by imposing condition {\em (ii)} of Theorem \ref{thm:HJM}.
	As a next step, for each $n=1,\ldots,N$, the random variable $\Delta A^{\delta}_{T_n}$ appearing in \eqref{eq:spread}, \eqref{eq:spread_FV} is determined by requiring that
	\beq	\label{eq:DeltaA}
	\Delta Y_{T_n}(T_n,\delta) = \delta\Delta L(T_n,T_n,\delta).
	\eeq
	Then, for each $n=1,\ldots,N-1$ and $T\in\{T_{n+1},\ldots,T_N\}$, the random variable $\Delta V(T_n,T,\delta)$ is determined by requiring that
	\beq	\label{eq:DeltaV}
	\Delta Y_{T_n}(T,\delta) = \delta\Delta L(T_n,T,\delta),
	\eeq
	while $\Delta V(T_n,T,\delta):=0$ for $T\leq T_n$. 
	Note that $\Delta V(T_n,T_{N+1},\delta) = 0$ for $\delta \neq 0$ and $n= 1,\ldots, N+1$.
	At this stage, the forward rates $\{f(\cdot,T,\delta):T\in\cT^{\delta}\}$ are completely specified.
	With this specification, it can be verified that conditions \eqref{eq:jumps_int_lib} and \eqref{eq:jump_sigma_int_lib} respectively imply that conditions \eqref{eq:jumps_int} and \eqref{eq:jump_sigma_int} of Theorem \ref{thm:HJM} are satisfied, using the fact that Assumption \ref{ass} as well as conditions \eqref{eq:jumps_int}, \eqref{eq:jump_sigma_int} are satisfied for $\delta=0$ and $T\in\cT^0$ by assumption.
	Moreover, it can be checked that, if condition {\em (ii)} of Theorem \ref{thm:Libor} is satisfied, then the random variables $\Delta A^{\delta}_{T_n}$ and $\Delta V(T_n,T,\delta)$ resulting from \eqref{eq:DeltaA}, \eqref{eq:DeltaV} satisfy conditions {\em (iii)}, {\em (iv)} of Theorem \ref{thm:HJM}, for every $n\in\N$ and $T\in\cT^{\delta}$.
	It remains to specify the process $\alpha^{\delta}$ appearing in \eqref{eq:spread_FV}. To this effect, an inspection of Lemma \ref{lem:bond_prices} and Corollary \ref{cor:bond_stoch_exp} reveals that, since the measure $\eta$ is purely atomic, the terms $f(t,t,\delta)$ and $f(t,t,0)$ do not appear in condition {\em (i)} of Theorem \ref{thm:HJM} and in condition \eqref{eq:HJM_r}, respectively. Since \eqref{eq:HJM_r} holds by assumption, $\alpha^{\delta}=0$ follows by imposing condition {\em (i)} of Theorem \ref{thm:HJM}.
	We have thus obtained that the two processes
	\[
	\bigl(1+\delta L(\cdot,T,\delta)\bigr)\frac{P(\cdot,T+\delta)}{X^0}
	\qquad\text{ and }\qquad
	\frac{S^{\delta}P(\cdot,T,\delta)}{X^0}
	\]
	are two local martingales starting from the same initial values, with the same continuous local martingale parts and with identical jumps. By means of \cite[Theorem I.4.18 and Corollary I.4.19]{JacodShiryaev}, we conclude that \eqref{eq:embed} holds for all $0\leq t\leq T\in\cT^{\delta}$.
\end{proof}

We want to point out that the specification described in Proposition \ref{prop:embed} is not the unique HJM model which allows embedding a given market model $\{L(\cdot,T,\delta):T\in\cT^{\delta}\}$. Indeed, $b(t,T_{i(t)},\delta)$ and $H^{\delta}_t$ can be arbitrarily specified as long as they satisfy 
\[
b(t,T_{i(t)},\delta) - H^{\delta}_t 
= b(t,T_{i(t)},0) + b(t,T_{i(t)+1},0) - \delta\frac{b^L(t,T_{i(t)},\delta)}{1+\delta L(t-,T_{i(t)},\delta)},
\]
together with suitable integrability requirements. An analogous degree of freedom exists concerning the specification of the functions $g(t,x,T_{i(t)},\delta)$ and $L^{\delta}(t,x)$.
Note also that the random variable $\Delta A^{\delta}_{T_n}$ given in Proposition \ref{prop:embed} can be equivalently expressed as
\[
\Delta A^{\delta}_{T_n}
= \frac{1+\delta L(T_n,T_n,\delta)}{1+\delta L(T_{n-1},T_{n-1},\delta)}\frac{P(T_n,T_{n+1})}{P(T_{n-1},T_n)}-1,
\qquad\text{ for }n=1,\ldots,N.
\]

\vspace{2cm}
\addresseshere

\end{document}